\documentclass[12pt]{article}

\usepackage{amsfonts}
\usepackage{amsmath}
\usepackage{amssymb}
\usepackage{amsthm}
\usepackage{mathrsfs}

\usepackage{authblk}

\newcounter{mnotecount}[section]
\let\oldmarginpar\marginpar
\renewcommand\marginpar[1]{\-\oldmarginpar[\raggedleft\footnotesize #1]%
{\raggedright\footnotesize #1}}

\title{Power Law Inflation with Electromagnetism}

\author[ ]{Xianghui Luo}
\author[$\dag$]{James Isenberg}
\affil[$\dag$]{Department of Mathematics, University of Oregon}

\date{} % Activate to display a given date or no date (if empty),
\begin{document}
\maketitle

\begin{abstract}
We generalize Ringstr\"om's global future causal stability results \cite{Ringstrom09} for certain expanding cosmological solutions of the Einstein-scalar field equations to solutions of the Einstein-Maxwell-scalar field system. In particular, after noting that the power law inflationary spacetimes $(M^{n+1}, \hat{g}, \hat{\phi})$ considered by Ringstr\"om in \cite{Ringstrom09} are solutions of the Einstein-Maxwell-scalar field system (with exponential potential) as well as of the Einstein-scalar field system (with the same exponential potential), we consider (nonlinear) perturbations of initial data sets of these spacetimes which include electromagnetic perturbations as well as  gravitational and scalar perturbations. We show that if (as in \cite{Ringstrom09}) we focus on pairs of relatively scaled open sets $U_{R_0} \subset U_{4R_0}$ on an initial slice of $(M^{n+1}, \hat{g})$, and if we choose a set of perturbed data which on $ U_{4R_0}$ is sufficiently close to that of $(M^{n+1}, \hat{g},\hat{\phi},\hat{A}=0)$, then in the maximal globally hyperbolic spacetime development $(M^{n+1},g,\phi,A)$ of this data via the Einstein-Maxwell-scalar field equations, all causal geodesics emanating from $U_{R_0}$ are future complete (just as in $(M^{n+1}, \hat{g})$). We also verify the controlled future asymptotic behavior of the fields in the spacetime developments of the perturbed data sets.

\end{abstract}

\tableofcontents
\newtheorem{definition}{Definition}
\newtheorem{theorem}{Theorem}
\newtheorem{lemma}{Lemma}
\newtheorem{proposition}{Proposition}
{\theoremstyle{remark} \newtheorem*{remark}{Remark}}
\newpage
\allowdisplaybreaks
%%%%%%%%%%%%%%%%%%%%%%%%%%%%%%%%%%%%%%%%%%%%%%%%%%%%%%%%%%%%%%%%
\section{Introduction}
\label{Intro}

\subsection{Background}
The well-posedness of the Einstein vacuum equations, the Einstein-Maxwell equations, and other related field equation systems has been established for many years \cite{CB52, B-CB-I-Y}. Specifically, it is known that for any given set of initial data satisfying the constraint equations, there exists a unique solution to the Einstein (or Einstein-Maxwell, etc.) equations for some amount of  time to the future. Well-posedness also establishes that small perturbations to an initial data set only lead to small changes to the corresponding solution in finite time, and that if those changes to the data are confined to a subset of the initial hypersurface then the changes in the solution occur strictly in the domain of dependence of that subset. However, well-posedness gives no information about global-in-time behavior of the development.

One way to formulate this issue for Einstein's and other relativistic equations is in terms of \emph{global future causal stability} (``GFC-stability") which addresses the following question: Given a particular set of initial data for which the maximal globally hyperbolic development (``MGH-development") is future causally geodesically complete, if one makes small perturbations to that initial data set, is the resulting MGH-development also future causally geodesically complete? Note that while the property of well-posedness generally characterizes a PDE system together with all (or none) of its solutions, GFC-stability pertains to a particular solution  or family of solutions of the system. Note also that GFC-stability concerns the full nonlinear PDE system (Einstein, Einstein-scalar, etc.), not a linearization of the system, and in referring to ``perturbations to the initial data set",  we mean new data sets suitably near the original given set of initial data (which, for convenience here, we will refer to as the ``background" data set).

Global stability (GFC-stability or other related varieties) has been studied extensively for Einstein's theory, but has been established for only a very small number of solutions: The epic work of Christodoulou and Klainerman \cite{Christodoulou} (see also the later generalizations and simplifications in \cite{Bieri, Lindblad}) proves the global stability of Minkowski spacetime for the vacuum and Einstein-Maxwell equations, while that of Friedrich \cite{Friedrich} shows that the DeSitter spacetime is globally stable for Einstein's equations with a cosmological constant. More recently, Andersson and Moncrief \cite{Andersson} have proven that the Milne spacetimes\footnote{These expanding spacetimes are constructed by spatially compactifying the mass hyperboloids in Minkowski spacetime} are globally stable solutions of the vacuum equations.

One feature of the DeSitter and the Milne solutions  which makes it a bit easier to establish global stability for them is the fact that they are expanding solutions. In a rough sense, this property acts to inhibit the concentration of curvature, so that perturbations do not tend to lead to singularities forming. Hence in searching for solutions expected to be globally stable, one is led to consider expanding solutions.

Motivated both by this consideration and by the recent astrophysical evidence \cite{Jarosik} that our universe is likely expanding at an accelerated rate, Ringstr\"om has  recently shown that certain solutions of the Einstein-scalar field equations with accelerating expansion are GFC-stable. He does this for both exponentially expanding background spacetimes satisfying the Einstein-scalar field equations 
with fairly general scalar field potential functions $V(\phi)$ \cite{Ringstrom08}, and for power law expanding background spacetimes satisfying the Einstein-scalar field equations with a certain set of  exponentially-decaying scalar potential functions \cite{Ringstrom09}.

In this work, we show that the power law expanding solutions considered by Ringstr\"om in \cite{Ringstrom09} are globally stable with respect to the Einstein-\emph{Maxwell}-scalar field equations. By choosing the electromagnetic fields to vanish, we may consider Ringstr\"om's solutions from \cite{Ringstrom09} to be solutions of the Einstein-Maxwell-scalar system. To prove stability of these background solutions in the larger  PDE system of course requires us to allow the perturbation solutions to include non-vanishing electromagnetic fields. This paper shows that this can be done, and that GFC-stability holds. We note that in a recent paper \cite{Svedberg}, stability has been proven for an electromagnetic generalization of Ringstr\"om's results \cite{Ringstrom08} for exponentially expanding solutions.

A key feature of background solutions with sufficiently accelerated expansion is that the analysis can be strongly localized. This is because the entire future of a small subset in the initial hypersurface is determined completely by the initial data on a small neighborhood of that subset. Effectively then, the topology of the Cauchy slices of the background solutions being tested for stability and the topology of the perturbed solutions is irrelevant.

The general structure of our proof is very similar to that of \cite{Ringstrom09}: i)  localizing  the analysis to the development of data sets on open sets in the initial Cauchy slice, with the formal extension of such local data sets to spatial tori; (ii) establishing the well-posedness of the Cauchy problem for the field perturbations relative to the background solutions (with appropriate handling of the gauge choice);  iii) defining energy-type functionals for the perturbation fields and their derivatives, and (with the help of bootstrap assumptions) proving monotonicity estimates for them; iii) using the energy estimates together with bootstrap arguments to prove long-time existence, regularity, and global estimates for the MGH-developments  of the perturbed initial data; (iv) using the global estimates to analyze the dynamics of causal paths in the perturbed  spacetimes, and thereby verifying future causal completeness. The theme of this paper is showing that all of these steps work for the Einstein-Maxwell-scalar field equations (with exponentially-decaying scalar field potentials).

%%%%%%
\subsection{Field Equations and Background Solutions}
\label{FE&BS}
Before stating our main results and  proceeding to prove them, we wish to set up the field equations for the parametrized set of Einstein-Maxwell-scalar field theories which we work with here, and we wish to also state what the background solutions are, explicitly.

The field variables for the Einstein-Maxwell-scalar field theories include the spacetime metric $g$, the electromagnetic vector potential $A$, and the scalar field $\phi$. Letting $R_{\mu\nu}$ and $R$ denote the Ricci tensor and the scalar curvature for $g$, letting $F$ denote the electromagnetic tensor for $A$, and choosing the scalar field potential to take the form $V(\phi) = V_0 e^{-\lambda \phi}$ (for constants $V_0$ and $\lambda$), we can write the field equations for this theory (for $(n+1)$-dimensional spacetimes) in the following (index) form:
\begin{align}
\label{fe1}R_{\mu\nu} - \frac{1}{2} R g_{\mu\nu} &= T_{\mu\nu},\\
\label{fe3}\nabla^\mu \nabla_\mu \phi - V'(\phi) &= 0,\\
\label{fe2}\nabla^\mu F_{\mu\nu} &= 0.
\end{align}
Here the stress-energy tensor for this system is given by
\begin{equation}
\label{Tmunu}
T_{\mu\nu} =  \partial_{\mu} \phi  \partial_{\nu} \phi - g_{\mu\nu} ( \frac{1}{2} g^{\rho\sigma} \partial_{\rho} \phi \partial_{\sigma} \phi + V(\phi) ) + ( F_{\mu\sigma} F_\nu\,^\sigma - \frac{1}{4} g_{\mu\nu} F_{\rho \sigma} F^{\rho \sigma}),
\end{equation}
and $V'(\phi)= -\lambda V(\phi)$.
Note that (\ref{fe1}) can be rewritten as
\begin{equation}
\label{fe1'}
R_{\mu\nu} = \partial_\mu \phi \partial_\nu \phi + \frac{2}{n-1} V(\phi) g_{\mu\nu} + F_{\mu\sigma} F_\nu\,^\sigma - \frac{1}{2(n-1)} g_{\mu\nu} F_{\rho\sigma} F^{\rho\sigma}.
\end{equation}
Note also that three parameters characterize these systems of field equations: the spatial dimension $n\ge 3$, the scalar potential scale $V_0>0$ and the scalar potential decay coefficient $\lambda>0$; hence for convenience, we shall denote  a particular choice of these theories by ``Einstein-Maxwell-scalar$_{\{n,V_0, \lambda \}}$".
Finally, note that the form of the field equations (\ref{fe1})-(\ref{Tmunu}) is consistent with the assumption here that the scalar fields are not charged, and so the interaction between the electromagnetic and the scalar fields is indirect (through the gravitational fields).

We now wish to specify the background fields $(\hat{g}, \hat{\phi}, \hat{A})$, which i) are solutions of the system (\ref{fe1})-(\ref{Tmunu}), ii) have accelerating expansion, and iii) (as we shall show) are GFC-stable. The fields are defined on the manifold $M^{n+1}= \mathbf{T}^n \times \mathbf{R}_+$, on which we choose the time coordinate $t>0$ and the global periodic spatial coordinates $x^i$. If we now choose the constant parameters  $t_0>0, p>1, c_0$, and $\kappa$, we write the following:

\begin{align}
\label{bgm} \hat{g} = - dt^2 + e^{2\kappa}(t/t_0)^{2p} \delta_{ij} dx^i dx^j,\\
\label{bg1} \hat{\phi} = \frac{2}{\lambda} \ln t - \frac{c_0}{\lambda} ,\\
\label{bg2} \hat{A}_{\mu} = 0 .
\end{align}

These fields do not generally satisfy the field equations (\ref{fe1})-(\ref{fe2}). However, if we require the field equation parameters $\{n, V_0, \lambda\}$ and the field parameters $\{ t_0, p, c_0,\kappa \}$ to satisfy the constraining relations
\begin{align}
\label{lambda}\lambda &= \frac{2}{[(n-1)p]^{1/2}} ,\\
\label{c0}c_0 &= \ln \left[ \frac{(n-1)(np - 1) p}{2 V_0} \right],
\end{align}
then indeed the fields $(\hat{g}, \hat{\phi}, \hat{A})$ do constitute a solution. Note that for a fixed spatial dimension $n$, (\ref{lambda}) expresses a one-to-one correspondence between the solution coefficient of expansion $p$, and the scalar potential exponent $\lambda$. So in effect, once one fixes the three field equation parameters $\{n, V_0, \lambda\}$, there remains a two parameter family of these background solutions. Note also that these solutions are identical to those appearing in \cite{Ringstrom09}, with the simple addition of the condition (\ref{bg2}). For convenience, we shall denote by ``$(\hat{g}, \hat{\phi}, \hat{A})_{ \{ t_0, p, c_0,\kappa \}}$"  a particular choice of the background solution to the Einstein-Maxwell-scalar$_{\{n, V_0, \lambda\}}$ field theory; in using this notation, we presume that the conditions (\ref{lambda})-(\ref{c0}) hold.

One of the key properties of any of the background spacetimes $(\mathbf{T}^n \times \mathbf{R}_+, \hat{g})_{ \{ t_0, p, c_0,\kappa \}} $ corresponding to the solutions (\ref{bgm})-(\ref{bg2}) is the accelerated expansion they exhibit, and the somewhat peculiar causal structure which consequently characterizes them. In particular, one finds that if one fixes a time $t_0$ and the corresponding Cauchy slice $\mathbf{T}^n_{t_0}$ in a background spacetime with expansion parameter $p$, and if for any point $q \in \mathbf{T}^n_{t_0}$ one considers a pair of coordinate balls  $B_{e^{-\kappa}\ell_0}(q)$ and $B_{3e^{-\kappa}\ell_0}(q)$ in $\mathbf{T}^n_{t_0}$  for the characteristic length  $\ell_0:=\frac{t_0}{p-1}$, then the causal future of $B_{e^{-\kappa}\ell_0}(q)$ is contained in the future domain of dependence of $B_{3e^{-\kappa}\ell_0}(q)$; in terms of standard notation (see, e.g., Wald \cite{Wald}), one has
\begin{equation}
\label{causal}
J^+ [ B_{e^{-\kappa}\ell(t_0)}(q) \times \{t_0\} ] \subseteq D^+ [B_{3e^{-\kappa}\ell(t_0)}(q)  \times \{t_0\}].
\end{equation}

The basis for this result is the fact that, if one considers any future causal path with starting point $(q,t_0)$ on the $t_0$ Cauchy surface $\mathbf{T}^n_{t_0}$, and if one calculates the projected spatial distance (relative to the induced metric) that the path can stray from $q$ on $\mathbf{T}^n_{t_0}$, a straightforward  calculation\footnote{It follows from the expression (\ref{bgm}) for the metric that any causal path $\gamma(s)$ satisfies the condition $ -( \dot{\gamma}^t)^2 + e^{2\kappa} (\frac{t}{t_0})^{2p}\delta_{ij} \dot{\gamma}^i\dot{\gamma}^j \le 0,$ which  can be rewritten as $e^{2\kappa} \delta_{ij}\dot{\gamma}^i \dot{\gamma}^j \le (\frac{t_0}{t})^{2p}(\dot{\gamma}^t)^2$. One then calculates the projected displacement as $\int^{s_1}_{s_0} [e^{2\kappa} \delta_{ij}\dot{\gamma}^i \dot{\gamma}^j]^{1/2}ds \le\int^{t_1}_{t_0} (\frac{t_0}{t})^pdt =\frac{t_0^p}{1-p} t^{1-p}|^{t_1}_{t_0}\le \ell_0. $ }
(see also \cite{Ringstrom09}) shows that it is bounded from above by $\ell_0$. Hence no causal path starting inside  $B_{e^{-\kappa}\ell_0}(q)$ can reach a spacetime point for which there are inextendible past directed paths which avoid  $B_{3e^{-\kappa}\ell_0}(q)$; the result follows.

Relying on this result, we can spatially localize the study of the GFC-stability of our background spacetimes, since in analyzing the future causal behavior of the spacetime evolved from perturbed data in $B_{e^{-\kappa}\ell_0}(q)$, we need not consider the influence of the development of any data outside of $B_{3e^{-\kappa}\ell_0}(q)$. Note that there is a small simplification of the proof of our results below if we work with an exterior ball of radius $4e^{-\kappa}\ell_0$ rather than $3e^{-\kappa}\ell_0$. Also, it is convenient to rescale the spatial metric in our background solutions by choosing the constant $\kappa= \kappa_0 := \ln [4\ell(t_0)]$. Doing this, we have the slightly simpler causal condition $J^+ [ B_{\frac{1}{4}}(q) \times \{t_0\} ] \subseteq D^+ [B_1(q)  \times \{t_0\}]$, which we can exploit in stating and proving our results here.

%%%%%
\subsection{Initial Value Formulation of the Field Equations}

The statement of our results, as well as the proof, rely strongly on a formulation of the field equations (\ref{fe1})-(\ref{Tmunu}) as an initial value problem. The standard $n+1$ ADM-type initial value formulation is as follows: The initial data consist of a choice of (i) a spatial manifold $\Sigma^n$, (ii) a Riemannian metric $h_{ab}$ and a symmetric tensor $K_{cd}$ on $\Sigma^n$ which together comprise the gravitational initial data, (iii) a pair of scalar fields $\varphi$ and $\pi$ on $\Sigma^n$ which provide initial data for the scalar field, and (iv) a two-form $B $ and one-form $E$ on $\Sigma^n$ which make up the electromagnetic initial data. The initial data set $(\Sigma^n, h, K, \varphi, \pi, B, E)$ satisfies the  \emph{Einstein-Maxwell-scalar constraint equations} (consisting of certain components of the field equations (\ref{fe1})-(\ref{Tmunu})) if the following hold \footnote{Here and throughout the paper, Latin indices run from $1$ to $n$ (space only) while Greek indices run from $0$ to $n$ (space plus time).}

\begin{align}
\label{cstrt1}
R - K_{ij}K^{ij} + (\mathrm{tr} K)^2 = \pi^2 + \nabla^i \varphi \nabla_i \varphi + 2 V(\varphi) + (E_j E ^j + \frac{1}{2}B_{ij}B^{ij}),\\
\label{cstrt2}\nabla^j K_{ji} - \nabla_i(\mathrm{tr} K) = \varphi \nabla_i \varphi + E_j B^j\,_i\,,\\
\label{cstrt3}\nabla_i  E^i = 0.
\end{align}
Here $\nabla$ is the Levi-Civita connection of $h$, $R$ is its scalar curvature, and the indices are raised and lowered using $h$.

Note that if we choose one of the natural $t=const.$ Cauchy surfaces (say, $t=t_1$) of the background solution $(\hat{g}, \hat{\phi}, \hat{A})_{ \{ t_0, p, c_0,\kappa \}}$, then the initial data on this Cauchy surface is $\hat{h}=e^{2\kappa}(\frac{t_1}{t_0})^{2p} \delta_{ij} dx^i dx^j$, $\hat{K}= -pe^{2\kappa}t_0^{-1}(\frac{t_1}{t_0})^{2p-1} \delta_{ij} dx^i dx^j$,  $\hat{\varphi}=\frac{2}{\lambda} \ln t_1-\frac{c_0}{\lambda}$, $\pi=\frac{2}{\lambda}\frac{1}{t_1}$, $\hat{B}=0$, and $\hat{E}=0$.

Given a choice of initial data satisfying the constraint equations, one seeks a \emph{globally hyperbolic development} of the data, which is a set $(M^{n+1}, g, \phi, A)$ such that the following hold true: (a) $(M^{n+1},g)$ is a globally hyperbolic spacetime, with $M^{n+1}$ diffeomorphic to $\Sigma^n \times \mathbf{R_+}$; (b) $(M^{n+1}, g, \phi, A)$ satisfies the Einstein-Maxwell-scalar field equations (\ref{fe1})-(\ref{Tmunu}); (c) there exists an embedding $i: \Sigma^n \to M^{n+1}$ such that $i(\Sigma^n)$ is a Cauchy hypersurface for $(M^{n+1}, g)$, with first and second fundamental forms $h$ and $K$, with $ \phi \circ i = \varphi$ and $\nabla_{e_\perp} \phi \circ i=\pi$ (for $e_\perp$ the future-directed unit normal vector field on $i(\Sigma^n)$), and with $B=i^*F$ and $E=i^*F(e_\perp,\cdot)$ for $F=dA$.

With small modifications (to generalize from the vacuum Einstein equations to the Einstein-Maxwell-scalar field equations), the well-known results of Choquet-Bruhat \cite{CB52} (see also \cite{B-CB-I-Y}) and of Choquet-Bruhat and Geroch \cite{CB-G} guarantee that for any smooth set of initial data satisfying the constraint equations (\ref{cstrt1})-(\ref{cstrt3}) there exists a globally hyperbolic development; moreover, for such data there exists a \emph{maximal globally hyperbolic development} (MGH-development)  unique up to isometry, which is maximal in the sense of containment (with appropriate isometry map).  The existence and uniqueness of MGH-developments plays a crucial role in the statement of our results, and in the proof of GFC-stability.

While it is useful to state our main theorem (below) in terms of initial data sets of the form $(\Sigma^n, h, K, \varphi, \pi, B, E)$, in carrying out the proof of our results we are led to work with modified specifications of initial data sets, which include quantities such as  $g_{0i}$ and $1+g_{00}$. Inclusion of these quantities is closely tied with the need to control gauges in the analysis, as we see below.

%%%%%%%%%%%%%%%%%%%

\subsection{Main Results}
\label{MainResults}

The standard idea of a stability theorem is that one fixes a solution of the field equations, noting certain important properties of the solution, one considers certain classes of perturbations of the solution, and one shows that the properties of interest remain true for the perturbed solutions. As a consequence of the localized character of the causal structure of the expanding solutions under study here (see Section \ref{FE&BS}), following \cite{Ringstrom09} we state our main theorems here in a slightly different way (which effectively leads to slightly stronger results). We consider sets of initial data $(\Sigma^n, h, K, \varphi, \pi, B, E)$ for an Einstein-Maxwell-scalar$_{\{n, V_0, \lambda\}}$ field theory which \emph{in local regions} are small perturbations of local initial data for one of our background solutions, and proceed to prove that the future development of the data restricted to a somewhat smaller region has the desired properties (causal geodesic completeness, etc.). With our results stated this way, they apply to solutions which may only locally be a small perturbation of one of the $(\hat{g}, \hat{\phi}, \hat{A})_{ \{ t_0, p, c_0,\kappa \}}$ solutions, or may be a perturbation of one of them in one region, and a different one in another region. This allows the results to hold for solutions with unrestricted topologies (unlike the background solutions, which are assumed to have the topology $\mathbf{T}^n \times \mathbf{R}_+)$.

To measure the degree to which initial data for the perturbed solutions locally deviate from that of the background solutions, we need to work with a set of norms.
Since the proof here depends on control of these norms via energy functionals, we are led to work with Sobolev norms; since the analysis is essentially local, we work with local Sobolev norms. In particular, for an open set $U\subset \Sigma^n$ diffeomorphic to a ball in $\mathbf{R^n}$ and therefore covered by Euclidean coordinates $(x^1,...x^n)$,  for a tensor field $\Psi$ on $\Sigma^n$ with $x^j$ coordinate-basis components $\Psi^{i_1\cdots i_q}_{j_1\cdots j_r}$, and for a non-negative integer $m$, we work with  Sobolev norms defined as follows
\[ \| \Psi \|_{H^m(U)} = \left( \sum^n_{i_1, \cdots, i_q =1} \sum^n_{j_1, \cdots, j_r =1} \sum_{| \alpha| \le m} \int_{x(U)} |\partial^\alpha \Psi^{i_1\cdots i_q}_{j_1\cdots j_r} \circ x^{-1} |^2 dx^1\cdots dx^n \right)^{1/2}.\]
Here the collective multi-index notation $``\partial^\alpha"$ is used for the partial derivatives, all of which are calculated using the $x^j$ coordinate basis.

In comparing (locally) a given set of initial data $(\Sigma^n, h, K, \varphi, \pi, B, E)$ for a perturbed solution with the data of a background solution, it is useful to find  the ``closest" background solution for the comparison. Presuming that the parameters $n, V_0$, and $\lambda$ have been chosen---thereby fixing the field theory and also thereby fixing  (via (\ref{lambda}) and (\ref{c0})) $p$ and $c_0$---it remains to determine $t_0$ and $\kappa$. As discussed above, it is convenient to choose (as a scaling) $\kappa_0 := \ln [4\ell(t_0)]$. Hence, one needs only to determine $t_0$.

The idea for determining  the appropriate choice of $t_0$ is based on equation (\ref{bg1}), which (for a given $\lambda$ and $c_0$) gives the time dependence of the background scalar field $\hat{\phi}$. Roughly speaking, to determine $t_0$ one calculates from the given (perturbed) data a local average of the scalar field, and then setting $\phi(t_0)$ equal to this average and
inverting (\ref{bg1}), one obtains $t_0$. More precisely, one chooses an open set $U\subset \Sigma^n$, together with a diffeomorphism $\zeta: U \rightarrow B_1(0) \subset \mathbf{R}^n$. Then, one calculates $\langle \varphi \rangle := \frac{1}{\omega_n} \int_{B_1(0)} \varphi \circ\, \zeta^{-1}\,dx $, where $\omega_n$ is the volume of the unit ball in $\mathbf{R}^n$ with respect to the Euclidean metric. Finally, one sets
\begin{equation}
\label{t0}
t_0 := \mathrm{exp} \left[ \frac{1}{2}(\lambda\langle\varphi \rangle + c_0)\right].
\end{equation}

We note that for a given set of initial data,  this procedure for mapping to a comparison background solution depends only on the choice of the open set $U$ and on the choice of the map $\zeta:U \rightarrow B_1(0)$; once that choice is made, $t_0$ (and therefore the comparison solution) is uniquely determined, regardless of whether the initial data is indeed close to a background solution. To simplify the discussion below, we use the notation $\Theta_{\{U,\zeta\}} (\Sigma^n, h, K, \varphi, \pi, B, E)$ to denote  the map taking the indicated data set to $t_0$, as defined above.

We are now ready to state our main theorem:

\begin{theorem}
\label{maintheorem}
Let  $(\Sigma^n, h, K, \varphi, \pi, B, E)$ be a set of initial data satisfying the constraint equations (\ref{cstrt1})-(\ref{cstrt3}) for a fixed choice of the Einstein-Maxwell-scalar$_{\{n, V_0, \lambda\}}$ field theory. There exists an $\epsilon>0$ (depending only on $n$ and $p$) such that if for some open set $U\subset \Sigma^n$ and for some diffeomorphism $\zeta: U \rightarrow B_1(0) \subset \mathbf{R}^n$ the data satisfy the smallness condition
\begin{eqnarray}
\label{epsilon}
\| e^{-2\kappa_0}h - \delta \|_{H^{m_0+1}(U)} + \| e^{-2\kappa_0}t_0 K - p \delta \|_{H^{m_0}(U)}  \nonumber \\
+ \| \varphi - \langle \varphi \rangle \|_{H^{m_0+1}(U)} + \| t_0 \pi - \frac{2}{\lambda} \|_{H^{m_0}(U)}\nonumber\\
+\sum_i \| E_i\|_{H^{m_0}(U)} + \sum_{i,k}\| B_{ik} \|_{H^{m_0}(U)} \le \epsilon,
\end{eqnarray}
with $t_0 =\Theta_{\{U,\zeta\}} (\Sigma^n, h, K, \varphi, \pi, B, E)$, with $\kappa_0$ chosen as above,
and with $m_0$ the smallest integer satisfying $m_0 > n/2 +1$;
then the MGH-development $(M^{n+1}, g, \phi, A)$ of  $(\Sigma^n, h, K, \varphi, \pi, B, E)$
has the property that if $i: \Sigma^n \to M$ labels the embedding corresponding to the initial data, then all causal geodesics starting in $i\{\zeta^{-1}[B_{1/4}(0)]\}$ are future complete.
\end{theorem}

This theorem shows that for a data set which in a local region is sufficiently close to data for one of the background solutions, geodesic completeness holds in the development of the data on a specified subset of that local region. If this holds for data in a neighborhood of every point in $\Sigma^n$, then clearly the entire MGH-development of the data set is future geodesically complete. It follows as a special case that if one chooses a Cauchy surface $\mathbf{T}^n_{t_0}$ for one of the background solutions $(\hat{g}, \hat{\phi}, \hat{A})_{ \{ t_0, p, c_0,\kappa \}}$, and if one considers sufficiently small  Einstein-Maxwell-scalar$_{\{n, V_0, \lambda\}}$ field perturbations of the data on $\mathbf{T}^n_{t_0}$, then the MGH-development of that data is future geodesically complete; hence the solutions $(\hat{g}, \hat{\phi}, \hat{A})_{ \{ t_0, p, c_0,\kappa \}}$ are all GFC-stable.

One might ask if, in addition to the property of future geodesic completeness, the MGH-development of a set of perturbed data has the property that its fields in some sense approach those of the corresponding background solution. This is in fact the case, in a certain weak sense:

\begin{theorem}
\label{2ndtheorem}
Let  $(\Sigma^n, h, K, \varphi, \pi, B, E)$ be a set of initial data satisfying the constraint equations (\ref{cstrt1})-(\ref{cstrt3}) for a fixed choice of the Einstein-Maxwell-scalar$_{\{n, V_0, \lambda\}}$ field theory, and also satisfying the $\epsilon$-smallness condition (\ref{epsilon}) from Theorem \ref{maintheorem}. Let $(M^{n+1}, g, \phi, A)$ denote the MGH-development of this data.
There are constants $t_{-} \in(0, t_0), \, a>0$, and $\kappa_m>0$ for all non-negative integers $m$,  there is a smooth map  $\Psi: (t_-, \infty) \times B_{5/8}(0) \to M^{n+1}$ which is a diffeomorphism onto its image and satisfies $\Psi(t_0, q) = i \circ \zeta^{-1} (q)$ for $ q\in B_{5/8}(0)$, and there is a Riemannian metric $H_{ab}$ on  $B_{5/8}(0)$ such that the following are true:

All causal paths that start in $i\{\zeta^{-1}[B_{1/4}(0)]\}$ remain in $\mathrm{Image}\{\Psi\}$ for all of the future.

Letting  $\| \cdot \|_{C^m}$ denote the $C^m$ norm on $B_{5/8}(0)$, letting $(g,\phi, A)$ denote the pullback of the MGH-development fields via $\Psi$, and letting $(\hat{g}, \hat{\phi}, \hat{A})$ denote the corresponding background fields, we have, for $t \ge t_0$, the following decay estimates:
\begin{align}
\label{exp1}
\| \phi(t, \cdot) - \hat{\phi}(t) \|_{C^m} + \| (t\partial_t \phi (t, \cdot)) - (t\partial_t \hat{\phi})(t)\|_{C^m} \le \kappa_m (t/t_0)^{-a}\,,\\
\| E_i \|_{C^m} = \| \partial_i A_0 - \partial_0 A_i \|_{C^m} \le \kappa_m e^{\kappa_0} (t/t_0)^p \,(t/t_0)^{-1-a},\\
\| B_{ij} \|_{C^m} = \| \partial_i A_j - \partial_j A_i \|_{C^m} \le \kappa_m e^{2\kappa_0} (t/t_0)^{2p} \,(t/t_0)^{-1-a},\\
\| (1 + g_{00}) (t, \cdot) \|_{C^m} + \| t\partial_t g_{00}(t, \cdot) \|_{C^m} \le \kappa_m (t/t_0)^{-a}\,,\\
\| \frac{1}{t}g_{0i}(t, \cdot) - \frac{1}{(n-2)p +1}H^{jl} \gamma_{jil} \|_{C^m} + \| t\partial_t (\frac{1}{t}g_{0i}(t, \cdot)) \|_{C^m} \nonumber\\ \le \kappa_m (t/t_0)^{-a}\,,\\
\| (t/t_0)^{-2p} e^{-2\kappa_0} g_{ij}(t, \cdot) - H_{ij} \|_{C^m}\qquad\qquad\qquad\qquad\quad \nonumber\\ +\,
\| (t/t_0)^{-2p} e^{-2\kappa_0} (t\partial_t g_{ij} (t, \cdot)) - 2pH_{ij}\|_{C^m} \le \kappa_m (t/t_0)^{-a}\,,\\
\| (t/t_0)^{2p} e^{2\kappa_0} g^{ij}(t, \cdot) - H^{ij} \|_{C^m} \le \kappa_m (t/t_0)^{-a}\,,\\
\label{exp7}  \| (t/t_0)^{-2p} e^{-2\kappa_0} t\,K_{ij}(t, \cdot) - pH_{ij} \|_{C^m}\le \kappa_m (t/t_0)^{-a}\,.
\end{align}
Here $\gamma_{jil}$ are the (lowered index) Christoffel symbols for the metric $H$ on $B_{5/8}(0)$, and $K_{ij}$ is the (evolving) second fundamental form for the hypersurface $B_{5/8}(0) \times \{t\}$.

\end{theorem}

We remark that while most of these inequalities clearly indicate decay, two of them appear not to do so: the second and the third, involving electromagnetic fields. We note, however, that the electromagnetic fields are vector components relative to coordinates in which the metric is expanding. If one considers locally measured fields (factoring out the expansion), then these fields do decay.

We also remark that these decay results do \emph{not} show that the developments of the perturbed data sets decay to the original background metrics directly. Rather, one obtains decay \emph{only} if one adds a diffeomorphism, and also adds the fiducial metric $H$ on $B_{5/8}$. It may be that sharper and more direct decay rates can be proven. We do not pursue this question here.

Finally, we note that while the results we prove here are generalizations of Theorem 2 in \cite{Ringstrom09},  we have not gone on to prove a generalization of Ringstr\"om's Theorem 3, which applies his Theorem 2 to prove global stability for  a class of locally spatially homogeneous spacetimes. Such results likely could be obtained for spatially homogeneous spacetimes containing electromagnetic fields; we do not, however, consider that issue here.

%%%%%%%%%%%%%%%%%%%%%%%%%
\subsection{Outline of the Proof}

A key feature of the proofs of Theorems \ref{maintheorem} and \ref{2ndtheorem} is the spatial localizability of the analysis, noted above. This allows the analysis to be carried out independently on each open set $U$ satisfying the hypotheses of Theorem \ref{maintheorem}. However in order to avoid working on regions with free boundaries, it is useful to patch the data set on $U$ into a set of background solution data on $\mathbf{T}^n\setminus U$, and then study the development of the patched-together data on $\mathbf{T}^n$. The expansion behavior of the background solutions as well as their perturbations guarantee that the development of the data on sufficiently small subsets of $U$ is independent of the externally patched-in data; hence the proof can be done via analysis on these patched data sets.

In general, the patched initial data sets violate the constraints in an annular region around $U$. It is thus necessary to formulate a global stability analysis that works for initial data sets which violate the constraint equations. We begin to set up such an analysis in Section \ref{Reform}. To start, we modify the field equations by introducing gauge source functions $\mathcal{D}^\mu$ and $\mathcal{G}$, and using them to hyperbolize the field equations. The gauge source function $\mathcal{D}^\mu$ replaces the contracted Christoffel symbols of the perturbed unknown metric by that of the background metric we perturb around and therefore is related to the ``wave coordinate gauge", while $\mathcal{G}$, if it is zero, corresponds to the Lorentz gauge of electromagnetism. For handling a global stability problem, hyperbolization of the field equations is not enough. We also need to add to the field equations some correction quantities that are expressed in terms of the gauge source functions. The purpose of these correction terms is to partially decouple the field equations to linear order in the field perturbed, and to insert damping terms in these equations. These features are helpful in proving stability for the modified system. Since we are interested in global stability of the original field equations, we reformulate the equations in such a way that the gauge source functions satisfy a system of hyperbolic equations which admit zero as a solution, and so that it is possible to prepare initial data for the modified equations so that the gauge source functions and their first order time derivatives vanish initially. Thus, the gauge source functions vanish identically, and by proving global stability for the modified equations, we obtain global stability for our original field equations.

We express our reformulated equations as PDEs for variables which are essentially the differences between the perturbed fields and the corresponding background fields. More specifically, we work with
$u:= 1 + g_{00}$, $u_i:= g_{0i}$, $\psi:= \phi - \hat{\phi}$, $A_0$, and $A_i$, all of which vanish for the background solution, and also  $\gamma_{ij}:= (t/t_0)^{-2p} g_{ij}$ which is time independent for $g_{ij} = \hat{g}_{ij}$. We further make a change of time coordinate from $t$ to $\tau$, such that $t \partial_t = \partial_\tau$, by defining $\tau = \ln (t/t_0)$. The reason for doing this is to eliminate the time dependences of the background fields. Thus we obtain a system of equations (\ref{refm1}) - (\ref{refm6}) such that each equation is in the form of a hyperbolic equation with dissipation and dispersion, plus some extra terms. Note that in the modified system, the equations for $u$, $\psi$ and $A_0$ are decoupled to first order from the equations for the other field variables.

In Section \ref{energies}, we define a sequence of energy functionals for the field variables, and we specify the bootstrap assumptions which we use to prove global existence in $\tau$. The bootstrap assumptions state that for all $\tau \in I =[0,s)$ for some unspecified $s$, we have solutions to the field equations, and the energy functionals  are controlled by some small number $\epsilon$.  Note that the specific bootstrap assumptions we use here are not optimal; they are chosen because they are sufficient to carry out the global existence argument, and because they make the estimation of the nonlinear terms (following the algorithm for estimates introduced in \cite{Ringstrom08}) applicable in our case.

One of the key tools for proving global existence is the set of differential inequalities (\ref{ine1}) - (\ref{ine5}) for the energy functionals. We derive these in Section \ref{energies}, and we also show that as a consequence of the hierachical structure of the equations, the differential inequalities  exhibit an hierachical structure as well. Relying on these inequalities, we are able to show (in Section \ref{S:global}) that the bootstrap assumptions can be improved. That is, combining the the bootstrap assumptions with \eqref{ine1}-\eqref{ine5}, we can show that a more rigorous version of the bootstrap assumptions holds. Bootstrap improvement is then used (also in Section \ref{S:global}) to verify the ``open" portion of an ``open-closed" type argument which show that indeed, the interval $I$ on which solutions exist extends to $s=\infty$. We thereby prove global existence. 

The remaining work in proving our two theorems is first to show that  the geodesics which start from a subset of the domain $U \times \{t_0\}$ are complete (we do this in Section \ref{S:geodesic}), and then to verify the asymptotic expansions (i.e., the decay results) (\ref{asym1}) - (\ref{asym8}) for the fields in the MGH-development of our initial data, as stated in Theorem \ref{2ndtheorem}; we do this in Section \ref{S:asymptotic}. We conclude our proof of the main theorems in Section \ref{S:proof}, and make concluding remarks  in Section \ref{S:conclusion}.

%%%%%%%%%%%%%%%%%%%%%%%%%%%%%%%%%%%%%%%%%%%%%%%%%%%%%%%%%%%%%%%%%
\section{ Field Equation Reformulation}
\label{Reform}

As noted in Section \ref{Intro}, as a consequence of the accelerated expansion of the background solutions, we can carry out the analysis in spatially local regions, which for convenience are each patched into a set of data for a background solution on $\mathbf{T}^n$ (generally different for each local region). We defer discussion of the details of the patching to Section \ref{S:proof}. For now, we presume that the patching has been done, and that we are consequently working with a set of data on $\mathbf{T}^n$ 
which is (in an appropriate sense) a small perturbation of a set of background field data, but does
 not necessarily satisfy the constraints everywhere. The MGH-development of the patched data on $\mathbf{T}^n$ consists of fields defined on the spacetime manifold $\mathbf{T}^n \times I$ for some interval $I$. Working on $\mathbf{T}^n \times I$, and working with fields which are small perturbations of the background fields, we can always choose coordinates $(x^i, t)$ with $x^i$ global periodic spatial coordinates.

The aim of the reformulation of the field equations (\ref{fe1})-(\ref{Tmunu}) we carry out here is to replace them by alternative equations which are manifestly hyperbolic and lead to well-posedness even if the constraints are not satisfied, and also to make sure that the reformulated equations can be used to develop a set of energy functionals which are controlled in time and lead to global existence. The replacement equations are obtained by adding terms to equations (\ref{fe1})-(\ref{Tmunu}) which can be made to vanish via gauge choice. We carry out the reformulation in two steps. In the first step, we add gauge terms which result in equations which are manifestly hyperbolic for the fields $g_{\mu \nu}, \phi$ and $A_{\mu}$ (with components defined using the coordinate basis $(x^i, t)$  defined above). In the second step, we rewrite the field equations in terms of the variables $u, u_i, \gamma_{ij}, \psi, A_0,$ and $A_i$ in a certain semi-decoupled form which is very useful for the analysis.

%%%%%
\subsection{Reformulation I: Hyperbolization}
\label{RefI}

The gauge functions we use to carry out the first stage of reformulation of the field equations (\ref{fe1})-(\ref{Tmunu}) are
\begin{align}
\label{gauge}
\mathcal{D}^\mu &:= \hat{\Gamma}^\mu - \Gamma^\mu,\\
\mathcal{G} &:= \nabla^\mu A_\mu,
\end{align}
where $\Gamma^\mu:= \frac{1}{2} g^{\alpha \beta}g^{\mu \nu}(\partial_{\alpha} g_{\beta \nu} +\partial_{\beta} g_{\nu \alpha} - \partial_{\nu} g_{\alpha \beta})$ is the contracted Christoffel symbol for the metric $g_{\alpha \beta}$, where $\hat{\Gamma}^\mu := \Gamma^\mu(\hat{g}) = \frac{np}{t} \delta^\mu_0$ is the same for the background metric $\hat{g}$, and where $\nabla$ is the covariant derivative compatible with $g_{\alpha \beta}$. Note that $\mathcal{D}^\mu$ is not covariant. However, we define $\mathcal{D}_\nu:= g_{\nu \mu} \mathcal{D}^\mu$, and also $\nabla_\mu \mathcal{D}_\nu := \partial_\mu \mathcal{D}_\nu - \Gamma_{\mu\nu}^\gamma \mathcal{D}_\gamma$.

Using these gauge quantities, we define modified versions of the Ricci and Faraday tensors
\begin{align}
\bar{R}_{\mu\nu} &= R_{\mu\nu} + \nabla_{(\mu}\mathcal{D}_{\nu)},\\
\bar{F}_{\mu\nu} &= F_{\mu\nu} + g_{\mu\nu} (\mathcal{G} - \mathcal{D}_\gamma A^\gamma),
\end{align}
and we also define certain arrays of correction terms
\begin{align}
M^{[g]}_{\mu\nu} &:= \frac{2p}{t}\left[ \begin{array} {cc} - \mathcal{D}^0 & \mathcal{D}_i\\ \mathcal{D}_i & 0 \end{array} \right], \\
M^{[\phi]} &:= - g^{\mu\nu} \mathcal{D}_\mu \partial_\nu \phi ,\\
M^{[A]}_{\nu} &:= - g^{\alpha\beta}\mathcal{D}_\alpha F_{\beta\nu} + \frac{2p}{t} g_{0\nu} (\mathcal{D}_\gamma A^\gamma - \mathcal{G}).
\end{align}
We then construct the following set of ``gauge-modified"  Einstein-Maxwell-scalar$_{\{n, V_0, \lambda\}}$ field equations
\begin{align}
\label{me1} \bar{R}_{\mu\nu}  + M^{[g]}_{\mu\nu} & =\partial_\mu \phi \partial_\nu \phi + \frac{2}{n-1} V(\phi) g_{\mu\nu} \nonumber\\ &\quad+ F_{\mu\sigma} F_\nu\,^\sigma - \frac{1}{2(n-1)} g_{\mu\nu} F_{\rho\sigma} F^{\rho\sigma} ,\\
\label{me2} g^{\alpha\beta} \partial_\alpha \partial_\beta \phi - \Gamma^\mu \partial_\mu \phi - V'(\phi) + M^{[\phi]} &= 0,\\
\label{me3} \nabla^\mu \bar{F}_{\mu\nu} + M^{[A]}_\nu &= 0 ,
\end{align}
noting the following properties: (i) If the gauge quantities $\mathcal{D}^\mu$ and $\mathcal{G}$ vanish,  then this system (\ref{me1})-(\ref{me3}) is satisfied if and only if the Einstein-Maxwell-scalar$_{\{n, V_0, \lambda\}}$ system (\ref{fe1})-(\ref{Tmunu}) is satisfied. (ii) The gauge modified system (\ref{me1})-(\ref{me3}) is manifestly (second order) hyperbolic for the fields $(g, \phi, A)$.

In view of these two properties of the gauge-modified system, our goal now is to show that for any given set of initial data\footnote{We use the tildes here to denote the fact that the data on $\mathbf{T}^n$ has been constructed via patching of local data on $U\subset \Sigma^n$ to background data on  $\mathbf{T}^n \setminus U$.} 
  $(\tilde h, \tilde K, \tilde \varphi, \tilde \pi, \tilde B,\tilde  E)$ (on the torus $\mathbf{T}^n$) for the Einstein-Maxwell-scalar$_{\{n, V_0, \lambda\}}$ system (\ref{fe1})-(\ref{Tmunu}), there is a corresponding set of initial data $(g, \partial_t  g, \varphi, \partial_t \varphi, A, \partial_t A)$ for (\ref{me1})-(\ref{me3}) which enforces the condition that the gauge quantities $\mathcal{D}^\mu$ and $\mathcal{G}$ vanish for as long as the solution exists. We show in the next subsection that this can be done.

%%%%%%
\subsection{Initial Data and Gauge Choice}
\label{GaugeChoice}
There are many choices of the data  $(g, \partial_t  g, \varphi, \partial_t \varphi, A, \partial_t A)$ which are consistent with a given set of initial data $( \tilde{h}, \tilde{K}, \tilde{\varphi}, \tilde{\pi}, \tilde{B}, \tilde{E})$ specified on $\mathbf{T}^n$; this is a manifestation of the gauge freedom in the Einstein-Maxwell-scalar field system. The following choice restricts this freedom, and in so doing, it enforces the condition that the gauge functions $\mathcal{D}^\mu$ and $\mathcal{G}$ vanish at the initial time\footnote{Note that $\mathcal{D}_0=0$ corresponds to (\ref{init3}), $\mathcal{D}_l=0$ corresponds to (\ref{init4}), and $\mathcal{G}=0$ corresponds to the first part of (\ref{init7}).}
$t_0$:
\begin{align}
\label{init1}
g_{ij}(t_0, \cdot) = \tilde{h} (\partial_i, \partial_j) , \quad g_{00}(t_0, \cdot) = -1, \quad g_{0j}(t_0, \cdot) = 0,\\
\partial_t g_{ij}(t_0, \cdot) = 2 \tilde{K}(\partial_i, \partial_j),\\
\label{init3}
\partial_t g_{00}(t_0, \cdot) = 2\hat{\Gamma}^0(t_0, \cdot) - 2\,\mathrm{tr} \tilde{K},\\
\label{init4}
\partial_t g_{0l}(t_0, \cdot) =\left[ -\tilde{h}_{li} \hat{\Gamma}^i + \frac{1}{2}\tilde{h}^{ij}(2\partial_i \tilde{h}_{jl} - \partial_l \tilde{h}_{ij})\right](t_0, \cdot),\\
\phi(t_0, \cdot) = \tilde{\varphi}, \quad \partial_t\phi(t_0, \cdot) = \tilde{\pi},\\
A_0(t_0, \cdot) = 0,\quad A_i(t_0, \cdot) = \tilde{A}_i(x),\\
\label{init7}\
\partial_t A_0 (t_0, \cdot) = \tilde{h}^{ij} \partial_i A_j (t_0, \cdot),\quad \partial_t A_i (t_0, \cdot) = -\tilde{E}_i,
\end{align}
where $\tilde{A}_k(x)$ is a solution of $\partial_i \tilde{A}_j(x) - \partial_j \tilde{A}_i(x) = \tilde{B}_{ij}(x)$.\footnote{We choose an appropriate $\tilde{A}_i(x)$ in Section \ref{S:proof}.}

The vanishing of the first time derivatives of  $\mathcal{D}^\mu$ and $\mathcal{G}$ on (portions of) the initial surface follows not from further restrictions on the gauge choice, but rather from the constraint equations (\ref{cstrt1})-(\ref{cstrt3}). The constraints are not satisfied everywhere (recall the consequences of the patching of the initial data) but they do hold on a subset $S$ of the initial hypersurface. To see that $\partial_t \mathcal{D}^\mu$ and $\partial_t \mathcal{G}$ vanish on $S$, we first calculate the quantity $G_{\mu \nu} -T_{\mu \nu}$, \emph{assuming that the gauge-modified field equations hold}, rather than the original system (\ref{fe1})-(\ref{Tmunu}); we obtain
\begin{equation}
\label{eq1}
G_{\mu\nu} - T_{\mu\nu} = - \nabla_{(\mu}\mathcal{D}_{\nu)} + \frac{1}{2}(\nabla^\rho \mathcal{D}_\rho) g_{\mu\nu} - M^{[g]}_{\mu\nu} + \frac{1}{2}(g^{\alpha\beta} M^{[g]}_{\alpha\beta}) g_{\mu\nu}.
\end{equation}
If we now let $e_{\perp}$ be the unit normal to the initial surface and let $X$ be any vector tangent to the surface, and if we contract (\ref{eq1}) with $e_{\perp}$ and $X$ while assuming that the constraints (\ref{cstrt1})-(\ref{cstrt3}) hold on $S$, then we obtain
\begin{equation}
- \frac{1}{2} (e_{\perp})^\mu X^\nu (\partial_\mu \mathcal{D}_\nu + \partial_\nu \mathcal{D}_\mu) = 0.
\end{equation}
Setting $t=t_0$, and noting that since $\mathcal{D}_i(t_0,\cdot)=0$ we must have $X^\nu \partial_\nu \mathcal{D}_\mu (t_0,\cdot) = 0$, it follows that
\begin{equation}
\partial_t \mathcal{D}_i(t_0, \cdot) = 0 \mbox{ on } S \subseteq \mathbf{T}^n.
\end{equation}
Arguing similarly, but now contracting (\ref{eq1}) twice with $e_{\perp}$  while assuming that the constraints hold at $t_0$, we obtain
\begin{equation}
\partial_t \mathcal{D}_0(t_0, \cdot) = 0  \mbox{ on } S \subseteq \mathbf{T}^n.
\end{equation}
Finally, if we contract  (\ref{me3}) with $e_{\perp}$, then (setting $t=t_0$) we may use the constraint (\ref{cstrt3}) to argue that $e_{\perp}^\nu \nabla^\mu F_{\mu\nu} (t_0,\cdot) = 0$; combining this with the vanishing of $\mathcal{D}_\mu(t_0, \cdot)$, $\mathcal{G}(t_0, \cdot)$ and $\partial_t \mathcal{D}_\mu(t_0, \cdot)$ (as shown above), we are left with
\begin{equation}
\partial_t \mathcal{G}(t_0, \cdot) = 0 \mbox{ on } S \subseteq \mathbf{T}^n.
\end{equation}

Now that we have determined that for any choice of initial data we may choose a gauge so that $\mathcal{G}$,
$\mathcal{D}_\mu$, and their first time derivatives vanish on the subset $S$ of the initial slice, if we wish to show that
$\mathcal{G}$ and $\mathcal{D}_\mu$ vanish in a spacetime neighborhood of $S$ (within its domain of dependence) it is sufficient to show that these quantities satisfy homogeneous wave equations. To do this, we first take the divergence of
equation (\ref{eq1}), noting that
\begin{equation}
\nabla^\mu T_{\mu\nu} = - M^{[\phi]} \partial_\nu \phi - F_\nu\,^\sigma (M^{[A]}_\sigma + \partial_\sigma (\mathcal{G} - \mathcal{D}_\gamma A^\gamma)).
\end{equation}
It follows that, presuming the gauge-modified equations have smooth solutions, we obtain
\begin{equation}
\label{Dhyperbolic}
g^{\alpha\beta} \partial_\alpha \partial_\beta \mathcal{D}_\mu + Q_\mu\,^{\alpha\beta} \partial_\alpha \mathcal{D}_\beta + S_\mu\,^\nu \partial_\nu \mathcal{G} + V_\mu\,^\nu \mathcal{D}_\nu + W_\mu \mathcal{G} =0
\end{equation}
for smooth functions $Q_\mu\,^{\alpha\beta}$, $ S_\mu\,^\nu$, $V_\mu\,^\nu$ and $W_\mu$.
Similarly, taking the divergence of (\ref{me3}), we obtain
\begin{equation}
\label{Ghyperbolic}
g^{\alpha\beta} \partial_\alpha \partial_\beta \mathcal{G} + H^\mu \partial_\mu \mathcal{G} + I^{\alpha\beta}\partial_\alpha \mathcal{D}_\beta + J \mathcal{G} + L^\mu \mathcal{D}_\mu = 0.
\end{equation}
for smooth functions  $H^\mu$, $I^{\alpha\beta}$, $J$ and $L^\mu$. This pair of equations together constitutes the desired homogeneous hyperbolic system for  $\mathcal{G}$ and $\mathcal{D}_\mu$.

%%%%%%
\subsection{Reformulation II: Perturbation Variables and First Order Semi-Decoupling }
The system of PDEs (\ref{me1})-(\ref{me3}) for the field variables $(g_{\mu \nu}, A_\mu, \phi)$ together with the system (\ref{Dhyperbolic})-(\ref{Ghyperbolic}) for the gauge quantities $(\mathcal{D}_\mu, \mathcal{G})$ constitute a coupled PDE system which could be used to argue local existence of solutions. To be able to show long time existence, however, it is advantageous to replace the variables $(g_{\mu \nu}, A_\mu, \phi)$ by others which are closely related to perturbations of the background solution; following \cite{Ringstrom09}, we choose to work with $u := 1 + g_{00}, u_i := g_{0i},  \gamma_{ij} := (t/t_0)^{-2p} g_{ij}, \psi := \phi - \hat{\phi}, A_0$, and $A_i$.

In deriving (from (\ref{me1})-(\ref{me3})) the evolution PDEs for these new variables, we wish to segregate those terms 
which are linear in perturbations of the background solution (i.e., linear in $u, u_i,  \gamma_{ij} , \psi,  A_0$, and $A_i$) 
from those which are higher order. Doing this, we obtain the following
\begin{align}
\label{ref1} - g^{\mu\nu}\partial_\mu \partial_\nu u + \frac{(n+2)p}{t} \partial_t u + \frac{2p[n(p-1) +1]} {t^2} u \nonumber\\ - \frac{8}{\lambda t} \partial_t \psi - \frac{2\lambda p(np-1)}{t^2}\psi +\, \Delta'_{00} = 0,\\
\label{ref2} - g^{\mu\nu}\partial_\mu \partial_\nu u_i + \frac{np}{t} \partial_t u_i + \frac{p(n-2)(2p-1)}{t^2}u_i \nonumber\\ - \frac{4}{\lambda t} \partial_i \psi - \frac{2p}{t} g^{lm}\Gamma_{lim} + \, \Delta'_{0i} = 0,\\
\label{ref3} - g^{\mu\nu}\partial_\mu \partial_\nu \gamma_{ij} + \frac{np}{t} \partial_t \gamma_{ij} + \frac{2p}{t^2} \left( \lambda (np-1) \psi- u \right) \gamma_{ij} + \, \Delta'_{ij} =0,\\
\label{ref4} - g^{\mu\nu}\partial_\mu \partial_\nu \psi + \frac{np}{t} \partial_t \psi + \frac{2(np-1)}{t^2}\psi - \frac{2}{\lambda t^2} u + \Delta'_\psi = 0,\\
\label{ref5} - g^{\mu\nu}\partial_\mu \partial_\nu A_0 + \frac{(n+2)p}{t} \partial_t A_0 + \frac{np(2p - 1)}{t^2} A_0 + \Delta'_e = 0,\\
\label{ref6} - g^{\mu\nu}\partial_\mu \partial_\nu A_i + \frac{(n-2)p}{t} \partial_t A_i + \frac{2p}{t} \partial_i A_0 + \Delta'_{b,i} = 0.
\end{align}
Here, the quantities $\Delta'_{\mu\nu}$, $\Delta'_\psi$, $\Delta_e'$ and $\Delta_{b,i}'$ are all of quadratic order or higher in the quantities that vanish for the background solution  (such as $u, u_i, \psi, A_0, A_i$, their first derivatives, and the first derivatives of $\gamma_{ij}$)\footnote{We include the primes on these $\Delta'$ quantities here because below, we replace them by slightly changed $\Delta$ quantities without primes.}. More specifically, one can write 
\begin{align}
\Delta'_{00} &= \tilde{\Delta}_{00} - 2\left[ F_{0\sigma} F_0\,^\sigma - \frac{1}{2(n-1)} g_{00} F_{\alpha\beta} F^{\alpha\beta}\right], \\
\Delta'_{0i} &= \tilde{\Delta}_{0i} - 2\left[ F_{0\sigma} F_i\,^\sigma - \frac{1}{2(n-1)} g_{0i} F_{\alpha\beta} F^{\alpha\beta}\right], \\
\Delta'_{ij} &= \tilde{\Delta}_{ij} - 2\left(\frac{t}{t_0}\right)^{-2p}\left[ F_{i\sigma} F_j\,^\sigma - \frac{1}{2(n-1)} g_{ij} F_{\alpha\beta} F^{\alpha\beta}\right], \\
\quad \Delta'_\psi &= \tilde{\Delta}_\psi,\\
\Delta'_e &=  2\left[ ( g^{\mu\rho} \Gamma_{0\mu}^\sigma - \frac{p}{t} g^{\rho\sigma} ) \partial_\rho A_\sigma - \frac{p}{t} \partial_t A_0 - \frac{p}{t} u ( \hat{\Gamma}^\mu A_\mu - g^{\rho\sigma} \partial_\rho A_\sigma ) \right], \\
\Delta'_{b,i} &= 2\left[ \frac{p}{t} F_{0i}+ g^{\rho\sigma} \Gamma_{i\rho}^\mu \partial_\sigma A_\mu  - \frac{p}{t} g_{0i} ( \hat{\Gamma}^\mu A_\mu - g^{\rho\sigma} \partial_\rho A_\sigma )\right],
\end{align}
where $\tilde{\Delta}_{\mu\nu}$ and $\tilde{\Delta}_\psi$ are written out explicitly in the work of Ringstr\"om; see equations (51), (52), (55) and (56) in \cite{Ringstrom09}, together with equations (82)-(87) and (92)-(93) in \cite{Ringstrom08}.

We note that the parameters $n, p$, and $\lambda$ appear in (\ref{ref1})-(\ref{ref6}) because the new variables $(u, u_i,  \gamma_{ij},  \psi, A_0, A_i)$ are defined in terms of the background solution $(\hat{g}, \hat{\phi}, \hat{A})_{ \{ t_0, p, c_0,\kappa \}}$. We also note that the quantity $g^{\mu \nu}$ as well as other metric and Christoffel quantities appearing in (\ref{ref1})-(\ref{ref6}) may be viewed as functions of $u, u_i$, and $\gamma_{ij}$.

The calculations leading from the PDEs (\ref{me1})-(\ref{me3}) to (\ref{ref1})-(\ref{ref4}) are essentially the same as those done in proving Lemma 3 in \cite{Ringstrom09}. To derive (\ref{ref5}) and (\ref{ref6}), we calculate as follows: 
\begin{align}
0&= \nabla^\mu \bar{F}_{\mu\alpha} + M^{[A]}_\alpha \nonumber\\
&=g^{\rho\sigma} (\partial_\rho F_{\sigma\alpha} - \Gamma_{\sigma\rho}^\gamma F_{\gamma\alpha} - \Gamma_{\alpha\rho}^\gamma F_{\sigma\gamma} ) + \partial_\alpha (\nabla^\mu A_\mu) \nonumber\\ &\quad - \mathcal{D}^\mu \partial_\mu A_\alpha - A_\mu \partial_\alpha \mathcal{D}^\mu + \frac{2p}{t} g_{0\alpha} ( \hat{\Gamma}^\mu A_\mu - g^{\rho\sigma} \partial_\rho A_\sigma ) \nonumber\\
&= g^{\rho\sigma} \partial_\rho \partial_\sigma A_\alpha - g^{\rho\sigma} \partial_\alpha \partial_\rho A_\sigma - \Gamma^\mu(\partial_\mu A_\alpha - \partial_\alpha A_\mu) \nonumber\\ &\quad- g^{\rho\sigma} \Gamma_{\alpha\rho}^\mu \partial_\sigma A_\mu + g^{\rho\sigma} \Gamma_{\alpha\rho}^\mu \partial_\mu A_\sigma + \partial_\alpha (g^{\rho\sigma} \partial_\rho A_\sigma) \nonumber\\ &\quad- \partial_\alpha ( \Gamma^\mu A_\mu) - \mathcal{D}^\mu \partial_\mu A_\alpha - A_\mu \partial_\alpha \mathcal{D}^\mu + \frac{2p}{t} g_{0\alpha} ( \hat{\Gamma}^\mu A_\mu - g^{\rho\sigma} \partial_\rho A_\sigma )\nonumber\\
&=g^{\rho\sigma} \partial_\rho \partial_\sigma A_\alpha - \hat{\Gamma}^\mu \partial_\mu A_\alpha - A_\mu \partial_\alpha \hat{\Gamma}^\mu - g^{\rho\sigma} \Gamma_{\alpha\rho}^\mu \partial_\sigma A_\mu \nonumber\\ &\quad+ g^{\rho\sigma} \Gamma_{\alpha\rho}^\mu \partial_\mu A_\sigma + (\partial _\alpha g^{\rho\sigma}) \partial_\rho A_\sigma + \frac{2p}{t} g_{0\alpha} ( \hat{\Gamma}^\mu A_\mu - g^{\rho\sigma} \partial_\rho A_\sigma ) \nonumber\\
&= g^{\rho\sigma} \partial_\rho \partial_\sigma A_\alpha - \hat{\Gamma}^\mu \partial_\mu A_\alpha - A_\mu \partial_\alpha \hat{\Gamma}^\mu - 2g^{\rho\sigma} \Gamma_{\alpha\rho}^\mu \partial_\sigma A_\mu \nonumber\\ &\quad+ \frac{2p}{t} g_{0\alpha} ( \hat{\Gamma}^\mu A_\mu - g^{\rho\sigma} \partial_\rho A_\sigma ).
\end{align}
For $\alpha = 0$, we obtain
\begin{align}
0&= \nabla^\mu \bar{F}_{\mu0} + M^{[A]}_0\nonumber\\
&= g^{\rho\sigma} \partial_\rho \partial_\sigma A_0 - \hat{\Gamma}^\mu \partial_\mu A_0 - A_\mu \partial_0 \hat{\Gamma}^\mu - 2g^{\rho\sigma} \Gamma_{0\rho}^\mu \partial_\sigma A_\mu \nonumber\\ &\quad+ \frac{2p}{t} g_{00} ( \hat{\Gamma}^\mu A_\mu - g^{\rho\sigma} \partial_\rho A_\sigma )\nonumber\\
&= g^{\rho\sigma} \partial_\rho \partial_\sigma A_0 - \frac{np}{t} \partial_t A_0 + \frac{np}{t^2} A_0 - 2g^{\rho\sigma} \Gamma_{0\rho}^\mu \partial_\sigma A_\mu \nonumber\\
&\quad - \frac{2p}{t} \hat{\Gamma}^\mu A_\mu + \frac{2p}{t}g^{\rho\sigma}\partial_\rho A_\sigma + \frac{2p}{t}(1+g_{00}) ( \hat{\Gamma}^\mu A_\mu - g^{\rho\sigma} \partial_\rho A_\sigma )\nonumber\\
&= g^{\rho\sigma} \partial_\rho \partial_\sigma A_0 - \frac{np}{t} \partial_t A_0 - \frac{np(2p-1)}{t^2} A_0 \nonumber\\
&\quad+\left[- 2g^{\rho\sigma} \Gamma_{0\rho}^\mu \partial_\sigma A_\mu + \frac{2p}{t} g^{\rho\sigma}\partial_\rho A_\sigma + \frac{2p}{t} u ( \hat{\Gamma}^\mu A_\mu - g^{\rho\sigma} \partial_\rho A_\sigma )\right]\nonumber\\
&= g^{\rho\sigma} \partial_\rho \partial_\sigma A_0 - \frac{(n+2)p}{t} \partial_t A_0 - \frac{np(2p-1)}{t^2} A_0 \nonumber\\
&\quad+2\left[ (\frac{p}{t} g^{\rho\sigma} - g^{\mu\rho} \Gamma_{0\mu}^\sigma) \partial_\rho A_\sigma + \frac{p}{t} u ( \hat{\Gamma}^\mu A_\mu - g^{\rho\sigma} \partial_\rho A_\sigma ) + \frac{p}{t} \partial_t A_0\right].
\end{align}
For $\alpha = i$, we obtain
\begin{align}
0&= \nabla^\mu \bar{F}_{\mu i} + M^{[A]}_i\nonumber\\
&= g^{\rho\sigma} \partial_\rho \partial_\sigma A_i - \hat{\Gamma}^\mu \partial_\mu A_i - A_\mu \partial_i \hat{\Gamma}^\mu - 2g^{\rho\sigma} \Gamma_{i\rho}^\mu \partial_\sigma A_\mu \nonumber\\ &\quad+ \frac{2p}{t} g_{0i} ( \hat{\Gamma}^\mu A_\mu - g^{\rho\sigma} \partial_\rho A_\sigma )\nonumber\\
&= g^{\rho\sigma} \partial_\rho \partial_\sigma A_i - \frac{(n-2)p}{t} \partial_t A_i - \frac{2p}{t} \partial_i A_0 \nonumber\\ &\quad+ 2\left[-\frac{p}{t} F_{0i} +\frac{p}{t} g_{0i} ( \hat{\Gamma}^\mu A_\mu - g^{\rho\sigma} \partial_\rho A_\sigma )- g^{\rho\sigma} \Gamma_{i\rho}^\mu \partial_\sigma A_\mu \right].
\end{align}

The evolution PDEs (\ref{ref1})-(\ref{ref6}) involve a number of factors of $t^{-1}$. These can be conveniently removed by multiplying all of the equations (\ref{ref1})-(\ref{ref6}) by $t^2$,  by replacing $t$ by $\tau := \ln(t/t_0)$ (so that consequently $\partial_{\tau}= t \partial_t$), and by setting $\tau_0:=\ln(t_0)$; we then obtain the following system of PDEs, to be solved for $(u(x,\tau)=g_{00}(x, e^{\tau+\ln(t_0)})=g_{00}(x, e^{\tau+\tau_0}), u_i(x,\tau)= etc...)$:
\begin{align}
\label{refm1} \tilde{\Box}_g u + [(n+2)p - 1] \partial_\tau u + 2p[n(p-1) +1] u \nonumber\\- \frac{8}{\lambda} \partial_\tau \psi - 2\lambda p(np-1) \psi  + \Delta_{00} = 0,\\
\label{refm2} \tilde{\Box}_g u_i + (np - 1) \partial_\tau u_i + p(n-2)(2p-1) u_i \nonumber\\ - e^{\tau+\tau_0} [\frac{4}{\lambda} \partial_i \psi + 2p g^{lm}\Gamma_{lim}] + \Delta_{0i} = 0,\\
\label{refm3} \tilde{\Box}_g \gamma_{ij} + (np - 1) \partial_\tau \gamma_{ij} + 2p \left( \lambda (np-1) \psi- u \right) \gamma_{ij} + \Delta_{ij} =0,\\
\label{refm4} \tilde{\Box}_g \psi + (np - 1) \partial_\tau \psi + 2(np-1) \psi - \frac{2}{\lambda} u + \Delta_\psi = 0,\\
\label{refm5} \tilde{\Box}_g A_0 + [(n+2)p - 1] \partial_\tau A_0 + np(2p - 1)  A_0 + \Delta_e = 0,\\
\label{refm6} \tilde{\Box}_g A_i + [(n-2)p -1] \partial_\tau A_i + 2p e^{\tau+\tau_0} \partial_i A_0 + \Delta_{b,i} = 0.
\end{align}
Here we define the hyperbolic operator $\tilde{\Box}_g$ via
\[   
\tilde{\Box}_g := -g^{00}\partial^2_\tau - 2 e^{\tau+\tau_0}g^{0i}\partial_\tau \partial_i - e^{2(\tau+\tau_0)} g^{ij}\partial_i \partial_j ,
\]
and we define the nonlinear remainder terms $\Delta_{\mu\nu}$, $\Delta_\psi$, $\Delta_e$ and $\Delta_{b,i}$ via
\begin{align}
&\Delta_{00} := (1 + g^{00}) \partial_\tau u + e^{2(\tau+\tau_0)} \Delta'_{00},\\
&\Delta_{0i} := (1 + g^{00}) \partial_\tau u_i + e^{2(\tau+\tau_0)} \Delta'_{0i},\\
&\Delta_{ij} := (1 + g^{00}) \partial_\tau \gamma_{ij} + e^{2(\tau+\tau_0)} \Delta'_{ij},\\
&\Delta_\psi := (1 + g^{00}) \partial_\tau \psi + e^{2(\tau+\tau_0)} \Delta'_\psi,\\
&\Delta_e := (1 + g^{00}) \partial_\tau A_0 + e^{2(\tau+\tau_0)} \Delta'_e,\\
&\Delta_{b,i} := (1 + g^{00}) \partial_\tau A_i + e^{2(\tau+\tau_0)} \Delta'_{b,i}.
\end{align}

Our stability analysis in the rest of the paper focuses on the evolution PDEs (\ref{refm1})-(\ref{refm6}). We note that, \emph{if we ignore the $\Delta$ terms in these equations and also ignore (for the moment) the dependence of the wave operator $\tilde{\Box}_g$ on the metric}, then we have the following semi-decoupled setup: (i) Equation (\ref{refm5}) involves $A_0$ alone. (ii) Equation (\ref{refm6}) involves only $A_i$ and $A_0$.
(iii) Equations (\ref{refm1}) and (\ref{refm4}) together form a coupled system for $u$ and $\psi$, independent of the other variables. (iv) Equation (\ref{refm3}) involves only $\gamma_{ij}$ and $u$ and $\psi$. (v) Equation (\ref{refm2}) involves $u_i$ and $u$ and $\psi$ and $\gamma_{ij}$, but not the electromagnetic variables. This semi-decoupled structure plays an important role in the analysis we carry out below.

%%%%%%%%%%%%%%%%%%%%%%%%%%%%%%%%%%%%%%%%%%%%%%%%%%%%%%%%%%%%%%%%%
\section{Energy Functionals, Bootstrap Assumptions, Estimates, and Differential Inequalities }
\label{energies}
The key tool for proving global existence for solutions to a Cauchy problem for a hyperbolic PDE system is the set of energy functionals for the system. For a general (nonlinear) system, these functionals are neither canonically determined nor unique. However they may be obtained, the necessary properties are (i) that (perhaps assuming certain a priori conditions on the fields) their future evolution is bounded, and (ii) that they control appropriate norms of the field variables. In this section, we obtain energy functionals for the  PDE system (\ref{refm1})-(\ref{refm6}), we state our bootstrap assumptions and derive estimates to control nonlinear terms (such as those contained in the $\Delta$ terms) and finally we  derive the differential inequalities which, so long as the bootstrap assumptions hold, control the evolution of the energy functionals.

\subsection{The Energy Functionals and their Differential Inequalities for a Fixed Spacetime Metric}

The field equations (\ref{refm1})-(\ref{refm6}) all involve the differential operator $\tilde{\Box}_g$, which of course involves the metric $g$. By definition of the field variables  $u, u_i$, and $\gamma_{ij}$, the metric $g$ is closely tied to them, and consequently $g$ evolves with them. This fact (a key feature of Einstein's gravitational field equations) must be taken into account in setting up the energy functionals and verifying their evolution properties.

However, we start our discussion of the energy functionals by artificially decoupling the metric as it appears in $\tilde{\Box}_g$ and elsewhere in coefficients of (\ref{refm1})-(\ref{refm6}) from the evolving fields  $u, u_i$, and $\gamma_{ij}$. We do this, in this subsection, by fixing a (generally time dependent) spacetime metric $g$, basing $\tilde{\Box}_g$ and the other coefficients on this fixed $g$, and treating $u, u_i$, and $\gamma_{ij}$ as independent. We recouple g and the field variables in the next subsection, with the help of bootstrap assumptions.

In defining the energy functionals, it is useful to treat the field variables in blocks, according to the semi-decoupling of the (linearized) evolution equation, noted above. So we start by working with just $u$ and $\psi$, which we write collectively as the 2-vector 
\begin{eqnarray}
 \mathbf{u} =  \left[\begin{array}{c}u\\ \psi\end{array}\right].
\end{eqnarray}
The evolution PDEs (\ref{refm1}) and (\ref{refm4}) then take the form 
\begin{equation}
\label{boldueqn}
\tilde{\Box}_g \mathbf{u} + C \partial_{\tau} \mathbf{u}+ J \mathbf{u} + \mathbf{\Delta} = 0,
\end{equation}
where $J$ and $C$ are the constant matrices
\begin{eqnarray}
J = \left[\begin{array}{cc} 2p[n(p-1)+1]&-2\lambda p(np-1)\\-2/\lambda&2(np-1)\end{array}\right], \nonumber\\ \quad C =  \left[\begin{array}{cc}(n+2)p-1&-8/\lambda\\0&np-1\end{array}\right],  \nonumber
\end{eqnarray}
and $\mathbf{\Delta}$ is the vector of nonlinear terms
\begin{eqnarray}
\mathbf{\Delta} = \left[\begin{array}{c} \Delta_{00}\\ \Delta_\psi\end{array}\right].
\end{eqnarray}

Since the only effect of the electromagnetic field on the evolution of $\mathbf{u}$ is via the $\mathbf{\Delta}$ term, the form of the energy functionals we use for $\mathbf{u}$ and the calculation of their evolutions are formally very similar to what appears in Section 4 of Ringstr\"om in \cite{Ringstrom09}. Following that narrative, we first obtain a matrix $T$  (see equation (76) in \cite{Ringstrom09}) which diagonalizes the matrix $J$, and then setting $\tilde{\mathbf{u}} = T^{-1}\mathbf{u}$, $\tilde{\mathbf{\Delta}} = T^{-1}\mathbf{\Delta}$, $\tilde{J} = T^{-1}JT = diag\{\lambda_-, \lambda_+\}$ and $\tilde{C} = T^{-1}CT$, we have
\begin{equation}
\label{lp}
\tilde{\Box}_g \tilde{\mathbf{u}} + \tilde{C} \partial_{\tau}\tilde{\mathbf{u}} + \tilde{J} \tilde{\mathbf{u}} + \tilde{\mathbf{\Delta}} = 0.
\end{equation}

We next define the base energy functional we shall use for $\tilde{\mathbf{u}}$.  Letting $c_{\mathcal{LS}}$, $b_1$, and $b_2$ be any set of positive definite constants, and using the notation $\tilde{g}^{ij} = e^{2(\tau + \tau_0)} g^{ij}$, we define\footnote{Here and below, the subscripts $\mathcal{LS}$ refer to the fact that $\mathbf{u}$ involves the metric ``lapse function" as well as the scalar field.} 
\begin{equation} 
\label{E_LS}
\mathcal{E}_{\mathcal{LS}} [\tilde{\mathbf{u}}] := \frac{1}{2} \int_{\mathbf{T}^n} (- g^{00}\partial_\tau \tilde{\mathbf{u}}^t \partial_\tau \tilde{\mathbf{u}} + \tilde{g}^{ij} \partial_i \tilde{\mathbf{u}}^t \partial_j \tilde{\mathbf{u}} - 2c_{\mathcal{LS}} g^{00}\tilde{\mathbf{u}}^t \partial_\tau \tilde{\mathbf{u}} + b_1\tilde{u}^2 + b_2\tilde{\psi}^2 )dx, 
\end{equation}
where the superscript $t$ on $\tilde{\mathbf{u}}^t$ indicates the transpose of $\tilde{\mathbf{u}}$. We then verify the following:
\begin{lemma}
\label{lemma1}
Let g be a fixed Lorentz metric on the spacetime $\mathbf{T}^n \times I$ for some interval $I$, and let $\tilde{\mathbf{u}}$ be a solution to equation (\ref{lp}) on $\mathbf{T}^n \times I$ for some choice of the constants $p>1$ and $n \ge 3$ characterizing an Einstein-Maxwell-scalar field theory. There exist positive constants $\eta_{\mathcal{LS}}$, $\zeta_{\mathcal{LS}}$, $b_1$, $b_2$ and $c_{\mathcal{LS}}$ (depending on $n$ and $p$) such that if we define $\mathcal{E}_{\mathcal{LS}} $ using the last three of these constants (as in \eqref{E_LS}) and if we require 
\begin{equation}
\label{c1} 
| g^{00} + 1 | \le \eta_{\mathcal{LS}}, 
\end{equation}
then
\begin{equation} \label{bound1}
\mathcal{E}_{\mathcal{LS}} [\tilde{\mathbf{u}}]  \ge \zeta_{\mathcal{LS}} \int_{\mathbf{T}^n} (\partial_\tau \tilde{\mathbf{u}}^t \partial_\tau \tilde{\mathbf{u}} + \tilde{g}^{ij} \partial_i \tilde{\mathbf{u}}^t \partial_j \tilde{\mathbf{u}} + \tilde{\mathbf{u}}^t  \tilde{\mathbf{u}} )dx 
\end{equation}
and
\begin{equation}
\label{integralineq}
 \frac{d\mathcal{E}_{\mathcal{LS}} [\tilde{\mathbf{u}}]} {d\tau}  \le - 2\eta_{\mathcal{LS}} \mathcal{E}_{\mathcal{LS}} + \int_{\mathbf{T}^n} [-(\partial_\tau \tilde{\mathbf{u}}^t + c_{\mathcal{LS}} \tilde{\mathbf{u}}^t)\mathbf{\tilde\Delta} + \Delta_{{\mathcal{LS}} }[\tilde{\mathbf{u}}]] dx,
 \end{equation}
where $\Delta_{{\mathcal{LS}} }[\tilde{\mathbf{u}}]$ is a function quadratic in $\tilde{\mathbf{u}}$ and its derivatives, defined in equation (\ref{shit}) below.
\end{lemma}

\begin{proof}
Examining inequality (\ref{bound1}), and using the choice of $\eta_{\mathcal{LS}}$ to ensure that $g^{00}$ is close to $-1$, we see that so long as $c^2_{\mathcal{LS}} < b_i$ for $i = 1, 2$,  there must exist a constant $\zeta_{\mathcal{LS}}$ for which inequality (\ref{bound1}) holds. To achieve $c^2_{\mathcal{LS}} < b_i$, 
we set 
\begin{equation}
\label{b_1b_2}
b_1 = \lambda_- + c_{\mathcal{LS}} \tilde{C}_{11}, \quad b_2 = \lambda_+ + c_{\mathcal{LS}} \tilde{C}_{22},
\end{equation}
where $ \tilde{C}_{11}$ and  $\tilde{C}_{22}$ are components of the matrix $\tilde{C}$ defined above, and we note that so long as we choose  $c_{\mathcal{LS}}$ to be sufficiently small, this inequality follows. 

Now consider the time derivative of $\mathcal{E}_{\mathcal{LS}}  [\tilde{\mathbf{u}}]$; we calculate
\begin{align}
\label{Energyderiv}
\frac{d\mathcal{E}_{\mathcal{LS}} } {d\tau} = &\int_{\mathbf{T}^n} \{ - \frac{1}{2} \partial_\tau \tilde{\mathbf{u}}^t (\tilde{C} + \tilde{C}^t ) \partial_\tau \tilde{\mathbf{u}} - \partial_\tau \tilde{\mathbf{u}}^t \tilde{J} \tilde{\mathbf{u}} - \partial_\tau \tilde{\mathbf{u}}^t \tilde{\mathbf\Delta} \nonumber \ \\&- ((p-1) + c_{\mathcal{LS}}) \tilde{g}^{ij} \partial_i \tilde{\mathbf{u}}^t \partial_j \tilde{\mathbf{u}} + c_{\mathcal{LS}} | \partial_\tau \tilde{\mathbf{u}} |^2 - c_{\mathcal{LS}} \tilde{\mathbf{u}}^t \tilde{C} \partial_\tau \tilde{\mathbf{u}} - c_{\mathcal{LS}} \tilde{\mathbf{u}}^t \tilde{J} \tilde{\mathbf{u}}
\nonumber\\ &-c_{\mathcal{LS}}\tilde{\mathbf{u}}^t \tilde{\mathbf{\Delta}}+ b_1 \tilde{u} \partial_\tau \tilde{u} + b_2 \tilde{\psi} \partial_\tau \tilde{\psi} + \Delta_{{\mathcal{LS}} } [ \tilde{\mathbf{u}}] \} dx,
\end{align}
where 
\begin{align}
\label{shit} \Delta_{{\mathcal{LS}} }  [ \tilde{\mathbf{u}}] := &- c_{\mathcal{LS}} (g^{00} + 1) \partial_\tau \tilde{\mathbf{u}}^t \partial_\tau \tilde{\mathbf{u}} - 2c_{\mathcal{LS}} (\tilde{g}^{0i} \partial_i \tilde{\mathbf{u}}^t \partial_\tau \tilde{\mathbf{u}} + (\partial_i \tilde{g}^{0i}) \tilde{\mathbf{u}}^t \partial_\tau \tilde{\mathbf{u}})\nonumber\\ &-c_{\mathcal{LS}} (\partial_j \tilde{g}^{ij}) \partial_i \tilde{\mathbf{u}}^t \tilde{\mathbf{u}} - \frac{1}{2} \partial_\tau g^{00} \partial_\tau \tilde{\mathbf{u}}^t \partial_\tau \tilde{\mathbf{u}} + \left[ \frac{1}{2} \partial_\tau \tilde{g}^{ij} + (p-1) \tilde{g}^{ij} \right] \partial_i \tilde{\mathbf{u}}^t \partial_j \tilde{\mathbf{u}} \nonumber\\ &- \partial_i \tilde{g}^{0i} \partial_\tau \tilde{\mathbf{u}}^t \partial_\tau \tilde{\mathbf{u}} - \partial_j \tilde{g}^{ij} \partial_\tau \tilde{\mathbf{u}}^t \partial_i \tilde{\mathbf{u}} - c_{\mathcal{LS}} \partial_\tau g^{00}  \tilde{\mathbf{u}}^t \partial_\tau \tilde{\mathbf{u}}.
\end{align}
Note that the constant $``(p-1)"$ appearing in equation (\ref{Energyderiv}) is labeled as $H$ in the analogous equation in \cite{Ringstrom09}; here, to avoid introducing an extra constant, we leave it in the form $``(p-1)"$.
Using (\ref{b_1b_2}) for $b_1$ and $b_2$, and choosing $c_{\mathcal{LS}}$ to be sufficiently small, we now obtain
\begin{align}
\frac{d\mathcal{E}_{\mathcal{LS}}} {d\tau} = &\int_{\mathbf{T}^n} \{ - \frac{1}{2} \partial_\tau \tilde{\mathbf{u}}^t (\tilde{C} + \tilde{C}^t ) \partial_\tau \tilde{\mathbf{u}} + c_{\mathcal{LS}} | \partial_\tau \tilde{\mathbf{u}} |^2 - ( p-1 + c_{\mathcal{LS}}) \tilde{g}^{ij} \partial_i \tilde{\mathbf{u}}^t \partial_j \tilde{\mathbf{u}}\nonumber\\ &- c_{\mathcal{LS}} \left[ \lambda_- \tilde{u}^2 + \lambda_+ \tilde{\psi}^2\right] - c_{\mathcal{LS}} \left[ \tilde{C}_{21} \tilde{\psi} \partial_\tau \tilde{u} + \tilde{C}_{12} \tilde{u} \partial_\tau \tilde{\psi} \right]\} dx,\nonumber \\ &+ \int_{\mathbf{T}^n} [-(\partial_\tau \tilde{\mathbf{u}}^t + c_{\mathcal{LS}} \tilde{\mathbf{u}}^t)\mathbf{\tilde\Delta} + \Delta_{{\mathcal{LS}}}[\tilde{\mathbf{u}}]] dx.
\end{align}
It follows from Lemma 5 of \cite{Ringstrom09}, that the matrix $\tilde{C} + \tilde{C}^t$ is positive definite. Thus,  provided  $g^{00}$ is close enough to $-1$ and provided  $c_{\mathcal{LS}}$ is chosen to be small enough, we obtain (\ref{integralineq}) for some constant $\eta_{\mathcal{LS}}$.
\end{proof}

It is important to note that the differential inequality (\ref{integralineq}) is useful for controlling the evolution of $\mathcal{E}_{\mathcal{LS}}$ only if we can establish estimates for $\mathbf{\tilde\Delta}$ and $\Delta_{{\mathcal{LS}}}$. We do this below, using the bootstrap assumptions. 

Since we need to control derivatives of $\mathbf{u}$ as well as $\mathbf{u}$ itself, it is very useful to work with energy functionals which involve higher spatial derivatives. For that purpose, we define the following sequence of energy functionals (labeled by positive integers $m$)
\begin{equation} 
E_{\mathcal{LS},m }= \sum_{|\alpha| \le m} \mathcal{E}_{\mathcal{LS}}[\partial^\alpha \tilde{\mathbf{u}}], 
\end{equation}
where $\partial^{\alpha}$ indicates the usual multi-index spatial derivative, of order $|\alpha|$. It follows readily from this definition of $E_{\mathcal{LS},m }$ and from Lemma \ref{lemma1} that the following differential inequality holds:
\newtheorem{corollary}{Corollary}
\begin{corollary}
\label{cor1}
Presuming the hypotheses of Lemma \ref{lemma1}, $E_{\mathcal{LS},m }$ satisfies
\begin{align}
\label{integralineqk}
 \frac{dE_{\mathcal{LS},m }}{d\tau} \le &- 2\eta_{\mathcal{LS}} E_{\mathcal{LS},m }\nonumber\\ &+ \sum_{|\alpha| \le m} \int_{\mathbf{T}^n} [(\partial^\alpha\partial_\tau \tilde{\mathbf{u}}^t + c_{\mathcal{LS}} \partial^\alpha\tilde{\mathbf{u}}^t)(-\partial^\alpha \tilde{\mathbf{\Delta}} + [ \tilde{\Box}_g, \partial^\alpha] \tilde{\mathbf{u}}) + \Delta_{\mathcal{LS}}[\partial^\alpha\tilde{\mathbf{u}}]] dx.
 \end{align}
\end{corollary}

Thus far, we have developed a sequence of energy functionals only for the pair of field variables $u$ and $\psi$ (with artificially fixed $g$). To obtain a similar sequence of energy functionals for $u_i, \gamma_{ij}, A_0$, and $A_j$, again with the metric $g$ fixed, it is useful to work with solutions of a model PDE
\begin{equation}
\label{form2}
\tilde{\Box}_g v + \alpha \partial_\tau v + \beta v = F
\end{equation}
for the scalar function $v$,
where $\alpha>0$ and $\beta\ge0$ are constants, $g$ is a fixed Lorentz metric on the spacetime manifold $\mathbf{T}^n\times I$ (for $n\ge 3$), and $F$ is a fixed function on the spacetime.  We have the following:
\begin{lemma}
\label{vlemma}
Let $v$ be a solution of equation (\ref{form2}) on $\mathbf{T}^n\times I$, with $g, \alpha, \beta$ and $F$  as stated above. 
There are constants $\eta_c, \zeta > 0$ and $\gamma, \delta \ge 0$, depending on $\alpha$ and $\beta$, such that if the given metric satisfies 
\begin{equation}
\label{c2} 
| g^{00} + 1 | \le \eta_c, 
\end{equation}
and if we define the energy functional via
\begin{equation}
\label{E_V}
\mathcal{E}_{\mathcal{V}} [v] := \frac{1}{2} \int_{\mathbf{T}^n} (- g^{00} (\partial_\tau v)^2 + \tilde{g}^{ij} \partial_i v\partial_j v - 2\gamma g^{00}v \partial_\tau v + \delta v^2 )dx, 
\end{equation}
then $\mathcal{E}_{\mathcal{V}} [v]$ bounds the following quadratic integral
\begin{equation}
\label{bound}
\mathcal{E}_{\mathcal{V}}[v] \ge \zeta \int_{\mathbf{T}^n} [(\partial_\tau v)^2 + \tilde{g}^{ij} \partial_i v \partial_j v + \iota_\beta v^2] dx,
\end{equation}
where $\iota_\beta = 0$ if $\beta = 0$ and $\iota_\beta = 1$ otherwise, 
and $\mathcal{E}_{\mathcal{V}} [v]$ satisfies the differential inequality 
\begin{equation}
\label{monotonicity}
\frac{d\mathcal{E}_{\mathcal{V}}}{d\tau} \le - 2\eta_c \mathcal{E}_{\mathcal{V}} + \int_{\mathbf{T}^n} [(\partial_\tau v + \gamma v)F + \Delta_{\mathcal{V}} [v] ] dx,
\end{equation}
where $\Delta_{\mathcal{V}} [v]$ is given by (\ref{crap}) below. If $\beta =0$, then $\delta = \gamma = 0$.
\end{lemma}
\begin{proof}
There are two separate cases to consider, depending on whether $\beta$ vanishes or not in the model PDE (\ref{form2}). If $\beta>0$, then we see that if we choose $\gamma = \alpha/2$ and $\delta = \beta + \alpha^2/2$, and if we choose $\eta_c$ so that $g^{00}$ is close to $-1$, then the inequality (\ref{bound}) holds. Calculating the time derivative of 
$\mathcal{E}_{\mathcal{V}}[v]$, we get
\begin{align}
\label{timeder}
\frac{d \mathcal{E}_{\mathcal{V}}[v] }{d \tau} = &\int_{\mathbf{T}^n} \{ - (\alpha - \gamma) (\partial_\tau v)^2 + (\delta - \beta - \gamma \alpha)v \partial_\tau v - \beta\gamma v^2\nonumber\\ \nonumber &- ((p-1) + \gamma) \tilde{g}^{ij} \partial_i v \partial_j v + (\partial_\tau v + \gamma v) F + \Delta_{\mathcal{V}}[v] \} dx,\\
\end{align}
with
\begin{align}
\label{crap}
\Delta_{\mathcal{V}} [v] = &- \gamma (\partial_i \tilde{g}^{ij}) v \partial_j v - 2 \gamma (\partial_i \tilde{g}^{0i}) v \partial_\tau v - 2 \gamma \tilde{g}^{0i} \partial_i v \partial_\tau v - (\partial_i \tilde{g}^{0i}) (\partial_\tau v)^2\nonumber\\ &- (\partial_j \tilde{g}^{ij})\partial_i v \partial_\tau v - \frac{1}{2} (\partial_\tau g^{00}) (\partial_\tau v)^2 + \left( \frac{1}{2} \partial_\tau \tilde{g}^{ij} + (p-1) \tilde{g}^{ij}\right) \partial_i v \partial_j v\nonumber\\ &-\gamma \partial_\tau g^{00} v \partial_\tau v - \gamma (g^{00} +1) (\partial_\tau v)^2.
\end{align}
If we now substitute in the above formulas for $\gamma$ and $\delta$, we have
\begin{align}
\frac{d \mathcal{E}_{\mathcal{V}}[v] }{d \tau} = &- \frac{1}{2} \int_{\mathbf{T}^n} \{ \alpha (\partial_\tau v)^2 + \alpha\beta v^2 + (2(p-1) + \alpha) \tilde{g}^{ij} \partial_i v \partial_j v \} dx\nonumber \\ &+ \int_{\mathbf{T}^n} \{(\partial_\tau v + \gamma v) F + \Delta_{\mathcal{V}}[v] \} dx.\nonumber
\end{align}
Combining this with the inequality (\ref{bound}), we determine that for some constant $\eta_c$, the differential inequality (\ref{monotonicity}) holds.

For the $\beta=0$ case, we choose $\gamma=0$ and $\delta=0$. We then readily check that, (for $g^{00}$ sufficiently close to $-1$), the energy inequality (\ref{bound}) holds. As well, using the calculation (\ref{timeder}), we readily verify (for the $\beta=0$ case) the differential inequality (\ref{monotonicity}).
\end{proof}

As with the energy functionals for $\mathbf{u}$, discussed above, it is useful to proceed from $\mathcal{E}_{\mathcal{V}}$ to a sequence of energy functionals involving higher derivatives of $v$:\begin{equation}
E_{\mathcal{V},m}[v] := \sum_{|\alpha| \le m} \mathcal{E}_{\mathcal{V}} [\partial^\alpha v]. 
\end{equation}
One then proves the following differential inequality result:
\begin{corollary}
\label{cor2}
Presuming that the hypotheses of Lemma \ref{vlemma} hold, the sequence of higher order energy functionals $E_{\mathcal{V},m}$ satisfy the following inequalities:
\begin{equation}
\label{EkIneq} 
\frac{dE_{\mathcal{V},m}}{d\tau} \le - 2\eta_c E_{\mathcal{V},m} + \sum_{|\alpha| \le m} \int_{\mathbf{T}^n} \{(\partial_\tau \partial^\alpha v + \gamma\, \partial^\alpha v)(\partial^\alpha F + [\tilde{\Box}_g, \partial^\alpha] v) + \Delta_{\mathcal{V}} [\partial^\alpha v] \} dx. 
\end{equation}
\end{corollary}

\begin{proof}
If one differentiates equation (\ref{form2}) and then applies Lemma \ref{vlemma}, the corollary immediately follows.
\end{proof}

To obtain the energy functionals for the fields $u_i, \gamma_{ij}, A_0$, and $A_j$, we now manipulate the  evolution equations 
for these fields so that, with varying specifications of the function $F$ and of the constants $\alpha$ and $\beta$, these evolution equations (for \emph{each of the components} of $u_i, \gamma_{ij}$, etc.) match the form of the model equation (\ref{form2}). For present purposes, we presume that $g$ is a fixed Lorentz metric, satisfying the condition $| g^{00} + 1 | \le \eta$, with $\eta$ a constant to be determined below.

For $u_i$,  we work with equation (\ref{refm2}). If we set $\alpha = np -1 > 0$ and $\beta = p(n-2)(2p-1) > 0$, and if we set $F$ equal to the negative of all except the first three terms\footnote{We note that, according to  this construction, the function $F$ includes information---terms depending on $\gamma_{ij}$, on $A_i$, etc.---which is not known. The formal derivation of the differential inequalities of the form (\ref{integralineq}) or (\ref{integralineqk}) still works, however. }
in (\ref{refm2}), then we have an equation of the form (\ref{form2}) for each of the components of $u_i$. It then follows from Lemma \ref{vlemma} and Corollary \ref{cor2} that  there exists a set of  positive constants $\gamma_{\mathcal{SH}}$, $\delta_{\mathcal{SH}}$, $\eta_{\mathcal{SH}}$ and $\zeta_{\mathcal{SH}}$ such that if we set $\gamma=\gamma_{\mathcal{SH}}$ and $\delta=\delta_{\mathcal{SH}}$ in the expression (\ref{E_V})  for $\mathcal{E}_{\mathcal{V}}$ and if we then define
\begin{equation} 
E_{u_i,m} :=  \sum_{| \lambda| \le m} \mathcal{E}_{\mathcal{V}}[ \partial^\lambda u_i], 
\end{equation}
(for $m$ any non-negative integer) then the conclusions of Lemma \ref{vlemma}, including the bounding condition (\ref{bound}) (with $\zeta=\zeta_{\mathcal{SH}}$) holds for $E_{u_i,0}$,
and the differential inequality \eqref{EkIneq} (with appropriate choices of constants) holds for each $E_{u_i,m}$. Furthermore, if we set
\begin{equation}
E_{\mathcal{SH},m}:= \sum_i E_{u_i,m}
\end{equation}
then one  verifies that similar sorts of bounds and differential inequalities (accounting for the summation) hold for these quantities. 

For $\gamma_{ij}$, we work with equation (\ref{refm3}). In this case, we set $\alpha = np-1, \beta = 0$, and $F=-\Delta_{ij}-2p(\lambda(np-1)\psi-u)\gamma_{ij}$; we then have an equation of the form (\ref{form2}) for the components of $\gamma_{ij}$. Lemma \ref{vlemma} now implies that for appropriate choices of constants $\gamma_{\mathcal{M}}=0$, $\delta_{\mathcal{M}}=0$, $\eta_{\mathcal{M}}$ and $\zeta_{\mathcal{M}}$, a bounding condition of the form (\ref{bound}) and a differential inequality of the form (\ref{monotonicity}) hold for $\mathcal{E}_{\mathcal{V}}[\gamma_{ij}]$. 
 However, the most useful energy functionals for the metric field $\gamma_{ij}$ are
\begin{equation}
\label{Emetric} 
E_{{\mathcal{M}},m} := \frac{1}{2} \sum_{i,j}\sum_{| \lambda | \le m} \left( \mathcal{E}_{\mathcal{V}}[ \partial^\lambda \gamma_{ij}] + \int_{\mathbf{T}^n} e^{-2a\tau} a_\lambda ( \partial^\lambda \gamma_{ij})^2 dx \right), 
\end{equation}
with $a_\lambda = 0$ for $|\lambda|= 0$ and $a_\lambda= 1$ for $|\lambda| > 0$ (and with $a >0$ a constant to be determined below, by (\ref{con2})). Below in Section \ref{Bootstrap}, using bootstrap assumptions and consequent estimates, we obtain  a differential inequality for each $E_{{\mathcal{M}},m}$ which controls its evolution for all time. We also verify that the energies $E_{{\mathcal{M}},m}$ control Sobolev norms of $\gamma_{ij}$ and of its derivatives. The unfamiliar second (integral) term in the definition \eqref{Emetric}  is needed so that the energies $E_{{\mathcal{M}},m}$ do indeed control these quantities. We note that this term is also consistent with the condition that the sequence of energies  $E_{{\mathcal{M}},m}$ all vanish for fields which are the same as the background fields $(\hat{g}, \hat{\phi}, \hat{A})_{ \{ t_0, p, c_0,\kappa \}}$.

We proceed to the energy functionals for the electromagnetic fields. For the electromagnetic scalar potential $A_0$, we work with the evolution equation (\ref{refm5}). It takes the desired form if we set $\alpha=(n+2)p-1, \beta =np(2p-1)$, and $F=-\Delta_e$. It then follows from Lemma \ref{vlemma} that for appropriate choices of constants $\gamma_{\mathcal{SP}}$, $\delta_{\mathcal{SP}}$, $\eta_{\mathcal{SP}}$ and $\zeta_{\mathcal{SP}}$, a bounding condition of the form (\ref{bound}) and a differential inequality of the form (\ref{monotonicity}) hold for $\mathcal{E}_{\mathcal{V}}[A_0]$. 
 The energy functionals for the scalar potential of primary interest are 
\begin{equation} 
\label{EScalpot}
E_{{\mathcal{SP}},m} = \sum_{| \lambda| \le m} \mathcal{E}_{\mathcal{V}}[ \partial^\lambda A_0].
\end{equation}
We use bootstrap arguments below (in Section  \ref{Bootstrap}) to obtain  differential inequalities which control  $E_{{\mathcal{SP}},m}$, and consequently control $A_0$ and its derivatives.

Finally for $A_j$, the electromagnetic vector potential, we work with equation (\ref{refm6}). To match the form of equation (\ref{form2}), we set $\alpha = (n-2)p-1$, $\beta =0$, and $F= -\Delta_{b,i}-2pe^{\tau+\tau_0}\partial_iA_0$.  It then follows from Lemma \ref{vlemma} that we have the constants $\gamma_{\mathcal{VP}}=0$, $\delta_{\mathcal{VP}}=0$, $\eta_{\mathcal{VP}}$ and $\zeta_{\mathcal{VP}}$; the main energy functionals for $A_j$ are 
\begin{equation}
\label{EVectpot}
E_{{\mathcal{VP}},m} = \sum_i\sum_{| \lambda | \le m} \mathcal{E}_{\mathcal{V}}[ \partial^\lambda A_i].
\end{equation}
Corollary \ref{cor2} implies that inequalities similar to (\ref{EkIneq}) can be derived; control of $E_{{\mathcal{VP}},m}$ is established in Section  \ref{Bootstrap}. Note that due to the vanishing of $\gamma_{\mathcal{VP}}$ and $\delta_{\mathcal{VP}}$, $E_{{\mathcal{VP}},m}$ does not control $\| A_i \|_{H^m}$; since we do not need to estimate $\| A_i \|_{H^m}$ anywhere in this paper, we do not include  in the definition of $E_{{\mathcal{VP}},m}$ an extra term similar to the extra term in the definition of $E_{{\mathcal{M}},m}$.

To establish appropriate differential inequalities for the energy functionals\footnote{The subscript labels for these energy functionals  $E_{{\mathcal{LS}},km}, E_{{\mathcal{SH}},m}, E_{{\mathcal{M}},m}, E_{{\mathcal{SP}},m}$, and $E_{{\mathcal{VP}},m}$ stand for, respectively, $\mathcal{L}$apse and $\mathcal{S}$calar field, $\mathcal{SH}$ift, $\mathcal{M}$etric, $\mathcal{S}$calar $\mathcal{P}$otential, and $\mathcal{V}$ector $\mathcal{P}$otential.}
$E_{{\mathcal{LS}},m }, E_{{\mathcal{SH}},m}, E_{{\mathcal{M}},m}, E_{{\mathcal{SP}},m}$, and $E_{{\mathcal{VP}},m}$, and consequently to control their evolution and that of the fields  $u, \psi, u_j, \gamma_{ij}, A_0,$ and $A_j$, we need to know something about the behavior of the function $``F"$ corresponding to each of the fields.  
In the next  subsection, we use bootstrap assumptions to establish this knowledge. Bootstrap assumptions also play a role in controlling the nature of the evolving metric $g$ which appears in the differential operator $\tilde{\Box}_g$ as well as elsewhere in the  evolution equations (\ref{refm1})-(\ref{refm6}), once we restore the relation between $g$  and the evolving field variables $u, u_i$, and $\gamma_{ij}$. We also discuss this in the next subsection.

%%%%%%%%%
\subsection{Bootstrap Assumptions and the Consequent Estimates}
\label{Bootstrap}
The idea of bootstrap assumptions is that one assumes that the evolving fields satisfy certain conditions for $t \in I \subset \mathbf{R}^1$, one uses those conditions to prove that certain estimates consequently hold (on $I$), and then one uses these estimates together with the evolution equations to argue that the solutions (satisfying the bootstrap assumptions) exist for all time $t\in\mathbf{R}^1$.

We use two bootstraps assumptions here. The first (which, following Ringstr\"om, we call the ``primary bootstrap assumption") essentially says that the evolving metric stays Lorentzian for $t\in I$. The second (called the ``main bootstrap assumption" ) says that the energy functionals for the fields evolving from the perturbed initial data set remain small for $t\in I$. We now state these more precisely, and discuss their immediate consequences.

To state the primary bootstrap assumption, it is useful to work with the following notation: Let $g$ be a real valued $(n+1)\times (n+1)$ matrix, with components  $g_{\mu \nu}$ for $\mu \in \{i,0\}$ and $i\in\{1,...n\}$. We use $g_\flat$ to denote the $n \times n$ submatrix with components $g_{ij}$, and we write $g_\flat >0$ if this matrix is positive definite. We use $g^{-1}$ to denote the inverse of $g$ and we use $g^\sharp$ to denote the $n \times n$ submatrix with components $g^{ij}$, presuming that these inverses exist. As above, we let $u [g] := 1+g_{00}$. To simplify notation here,\footnote{The specific aim is to distinguish the vector $u_i$ from the scalar $u$ without using indices.} rather than using $u_i$ to denote $g_{0i}$ we write $v[g]:=(g_{01},\cdots, g_{0n})$. As well
we use $u[g^{-1}] := 1+g^{00}=1+(g_{00}-g^{ij}u_iu_j)^{-1}$ and $v[g^{-1}]:=(g^{01}, \cdots, g^{0n})=(\cdots, (g^{ij}u_iu_j-g_{00})^{-1}g^{kl}u_l, \cdots)$.  We recall that $g$ is defined to be a Lorentz metric if it is symmetric, has $n$ positive eigenvalues and a single  negative one. Noting that $g$ is Lorentzian if $u[g]<1$ and if $g_\flat>0$, we call $g$ a canonical Lorentz metric if these two inequalities hold. We use $\mathcal{L}_n$ to denote the set of canonical $(n+1)\times (n+1)$ Lorentz matrices. 

Now, given a symmetric positive definite $n\times n$ matrix $M$ with components $M_{ij}$, and given a vector  $w \in \mathbf{R}^n$, we write (presuming the Einstein summation convention) 
\[ 
| w |_M := \left( M_{ij} w^i w^j \right)^{1/2}. 
\]
If $M$ is the identity matrix, we simply write $| w | := | w |_{Id}$. We can now state the following:

\begin{definition}
\label{PrimaryBootstrap}
Let $p > 1$, $ a > 0$, $c_1 > 1$, $\eta \in (0, 1)$, $\tau_0$ and $\kappa_1$ be real numbers and let $n \ge 3$ be an integer. A function $g: I \times \mathbf{T}^n \rightarrow \mathcal{L}_n$ satisfies the \textbf{primary bootstrap assumption}  $\mathcal{PBA}_{\{p, a, c_1, \eta, \tau_0, \kappa_1, n\}}$ on an interval $I$  if
\begin{align}
\label{p0}\frac{1}{c_1} | w |^2 \le e^{-2p\tau - 2\kappa_0} | w |_{g_\flat}^2  \le  c_1 | w |^2,\\
\label{p1} | u[g] |  \le  \eta,\\
\label{p2} | v[g] |^2  \le  \eta c_1^{-1} e^{2p\tau + 2\kappa_0 - 2 a\tau},
\end{align}
for all $ w \in \mathbf{R}^n$ and for all $(\tau, x) \in I \times \mathbf{T}^n$, with $\kappa_0 = \tau_0 + \kappa_1$.
\end{definition}
We remark  that in Section \ref{S:global}, we determine $t_0$ from the initial data and we set $\kappa_0 := \ln [4\ell(t_0)] = \ln \frac{4t_0}{p-1}$. It then follows that $\kappa_1 = \ln \frac{4}{p-1}$.

It is useful to note the following immediate consequence of the primary bootstrap assumption (following Lemma 7 of \cite{Ringstrom08}), which we use to control the inverse of the metric:

\begin{lemma}
\label{inverse}
If $g: I \times \mathbf{T}^n \rightarrow \mathcal{L}_n$ satisfies the $\mathcal{PBA}_{\{p, a, c_1, \eta, \tau_0, \kappa_1, n\}}$ on an interval $I$, then there exists a constant $\eta_0 \in (0, 1/4)$ such that if $\eta \le \eta_0$, then\footnote{We note that $(v[g], v[g^{-1}]) =g_{0i}g^{0i}.$}
\begin{align}
| v[g^{-1}] | \le 2 c_1 e^{-2p\tau - 2\kappa_0} | v[g] |,\\
| (v[g], v[g^{-1}]) | \le 2 c_1 e^{-2p\tau - 2\kappa_0} | v [g] |^2,\\
|u[g^{-1}]| \le 4 \eta,\\
\frac{2}{3c_1} | w |^2 \le e^{2p\tau + 2\kappa_0} | w |^2_{g^\sharp} \le \frac{3c_1}{2} | w |^2,
\end{align}
hold for all $ w \in \mathbf{R}^n$ and all $(\tau, x) \in I \times \mathbf{T}^n$.
\end{lemma}

To show that the energy functionals we have defined, so long as they are finite, control the fields and their derivatives, we need only the primary bootstrap assumptions. To argue this, we fix a time interval $I$, and assume that  $\mathcal{PBA}_{\{p, a, c_1, \eta, \tau_0, \kappa_1, n\}}$ holds, with $\eta$ and $a$ set as
\begin{align}
\label{con1}
\eta &:= \min\{\eta_0, \eta_{\mathcal{LS}}/4, \eta_{\mathcal{SH}}/4, \eta_{\mathcal{M}}/4, \eta_{\mathcal{SP}}/4, \eta_{\mathcal{VP}}/4\},\\
\label{con2}
a &:= \frac{1}{4} \min\{p-1, \eta_{\mathcal{LS}}, \eta_{\mathcal{SH}}, \eta_{\mathcal{M}}, \eta_{\mathcal{SP}}, \eta_{\mathcal{VP}}\},
\end{align}
where $\eta_{\mathcal{LS}}$ comes from Lemma \ref{lemma1}, while $ \eta_{\mathcal{SH}}, \eta_{\mathcal{M}}, \eta_{\mathcal{SP}}$, and $\eta_{\mathcal{VP}}$ are the corresponding constants which arise in the application of Lemma \ref{vlemma} to the other fields. We note (as a consequence of Lemmas  \ref{lemma1} and \ref{vlemma}) that these constants $a$ and $\eta$ depend only on the dimension $n$ and the expansion parameter $p$. Invoking 
a convenient rescaling (which we justify immediately following the proof of Lemma \ref{Delta}) in defining the following,
\begin{align}
\label{scalings}
\tilde{E}_{\mathcal{LS},m} := e^{2a\tau} E_{\mathcal{LS}, m}, \, \tilde{E}_{\mathcal{SH},m} := e^{-2p\tau - 2\kappa_0 + 2a\tau} E_{\mathcal{SH}, m},\nonumber\\ \tilde{E}_{\mathcal{M},m} := e^{- 4\kappa_0+ 2a\tau} E_{\mathcal{M}, m},\, \tilde{E}_{\mathcal{SP},m} := e^{2a\tau} E_{\mathcal{SP}, m}, \, \tilde{E}_{\mathcal{VP},m} := e^{-2p\tau - 2\kappa_0 + 2a\tau} E_{\mathcal{VP}, m},
\end{align}
we obtain a lemma which summarizes the desired (Sobelev) norm control of the fields:
\begin{lemma}
\label{control lemma}
Let the fields $\{u, u_i, \gamma_{ij}, \psi, A_0, A_i\}$ be specified on the spacetime manifold,
$I \times \mathbf{T}^n$ for some time interval $I$. If the metric $g$ constructed from these fields satisfies the $\mathcal{PBA}\{p, a, c_1, \eta, \tau_0, \kappa_1,n \}$ (and therefore is Lorentzian) for $\eta$ and $a$  defined by (\ref{con1}) and (\ref{con2}), respectively, and for real numbers $c_1>1, \tau_0$, and $\kappa_1$, then the inequalities
\begin{align}
\label{51} e^{a\tau} [ \| \psi \|_{H^m} + \| \partial_\tau \psi \|_{H^m} + e^{-(p-1)\tau - \kappa_1} \| \partial_i \psi \|_{H^m} ] \le C \tilde{E}^{1/2}_{\mathcal{LS},m},\\
\label{eq11} e^{a\tau} [ \| u \|_{H^m} + \| \partial_\tau u \|_{H^m} + e^{-(p-1)\tau - \kappa_1} \| \partial_i u \|_{H^m} ] \le C \tilde{E}^{1/2}_{\mathcal{LS},m},\\
\label{eq10} e^{-p\tau - \kappa_0 + a\tau} [ \| u_j \|_{H^m} + \| \partial_\tau u_j \|_{H^m} + e^{-(p-1)\tau - \kappa_1} \| \partial_i u_j \|_{H^m} ] \le C \tilde{E}^{1/2}_{\mathcal{SH},m},\\
\label{es0} e^{-2p\tau - 2\kappa_0 + a\tau} [ \| \partial_\tau g_{ij} - 2p g_{ij} \|_{H^m} + e^{-(p-1)\tau - \kappa_1} \| \partial_l g_{ij} \|_{H^m} ] \le C \tilde{E}^{1/2}_{\mathcal{M},m},\\
\label{es} e^{-2p\tau - 2\kappa_0} \| \partial^\alpha g_{ij} \|_2 \le C \tilde{E}^{1/2}_{\mathcal{M},m},\,\, 0 < | \alpha | \le m,\\
\label{52} e^{a\tau} [ \| A_0 \|_{H^m} + \| \partial_\tau A_0 \|_{H^m} + e^{-(p-1)\tau - \kappa_1} \| \partial_i A_0 \|_{H^m} ] \le C \tilde{E}^{1/2}_{\mathcal{SP},m},\\
\label{53} e^{-p\tau - \kappa_0 + a\tau} [ \| \partial_\tau A_i \|_{H^m}  + e^{-(p-1)\tau - \kappa_1} \| \partial_l A_i \|_{H^m} ] \le C \tilde{E}^{1/2}_{\mathcal{VP},m}\,,
\end{align}
hold on $I$, where the constants depend on $c_1$, $n$ and $p$.
\end{lemma}

\begin{proof} This result follows from the definitions of the energy functionals, and from the application of Lemma \ref{lemma1} and Lemma \ref{vlemma}.  We note that in applying these lemmas, we only use the conclusions regarding energies and norms (e.g., \eqref{bound1}), not the conclusions regarding differential inequalities (e.g., \eqref{integralineq}). To reach these conclusions, it is necessary that the metric satisfy  $\mathcal{PBA}_{\{a, c_1, \eta, \tau_0, \kappa_1, n\}}$, but it is \emph{not} necessary that the fields satisfy the field equations \eqref{refm1}-\eqref{refm6}. 
\end{proof}

To state the main bootstrap assumption, it is useful to define the following collective $m^{th}$-order energy functional
\begin{equation}
\label{Ek}
\tilde{E}_m := \tilde{E}_{\mathcal{LS},m}  + \tilde{E}_{\mathcal{SH},m} + \tilde{E}_{\mathcal{M},m} + \tilde{E}_{\mathcal{SP},m} + \tilde{E}_{\mathcal{VP},m}.
\end{equation}
We can then define the following:
\begin{definition}
Let $p > 1$, $c_1 > 1$,  $\tau_0$, $\kappa_1$ and $\epsilon>0$ be real numbers, let $n \ge 3$ be an integer, and let $m_0 > n/2 +1$ be an integer.
Let $\eta$ and $a$ be specified by (\ref{con1}) and (\ref{con2}) respectively. The fields $(g, \psi, A)$ satisfy the  \textbf{main bootstrap assumption} $\mathcal{MBA}_{\{p, a, c_1, \eta,  \tau_0, \kappa_1, n,  \epsilon\}}$ on the interval $I$ if
\begin{description}
\item[1.]  $g: I \times \mathbf{T}^n \rightarrow \mathcal{L}_n$, $\psi: I \times \mathbf{T}^n \rightarrow \mathbf{R}$ and $A_\mu: I \times \mathbf{T}^n \rightarrow \mathbf{R}$ are $C^\infty$,
\item[2.]  $g$ satisfies $\mathcal{PBA}_{\{p, a, c_1, \eta, \tau_0, \kappa_1,n\}}$ on $I$,
\item[3.]  $g$, $\psi$ and $A$ satisfy
 \begin{equation} \label{energycontrol}\tilde{E}^{1/2}_{m_0}(\tau) \le \epsilon, \end{equation}
\end{description}
for $\tau \in I$.
\end{definition}

The primary use of the main bootstrap assumption is to control the (nonlinear) $\Delta$-terms which appear in the field equations \eqref{refm1}-\eqref{refm6}. Controlling these $\Delta$ terms, as well as certain commutator terms involving spacetime derivatives and the wave operator $\tilde \Box$, we can derive  the differential inequalities (\ref{ine1}) - (\ref{ine5}), which are crucial to our proof of global existence. We state the level of control of the $\Delta$-terms in the following lemma:

\begin{lemma}
\label{Delta}
If the fields  $(g, \psi, A)$ satisfy  $\mathcal{MBA}_{\{p, a, c_1, \eta,  \tau_0, \kappa_1, n,  \epsilon\}}$ on an interval $I$, then the following estimates 
\begin{align}
\label{old_est1} \| \Delta_{00} \|_{H^m} \le C \epsilon e^{-2a\tau} \tilde{E}^{1/2}_m\,,\\
\label{old_est2} \| \Delta_{0i} \|_{H^m} \le C \epsilon e^{p\tau + \kappa_0 -2a\tau } \tilde{E}^{1/2}_m\,,\\
\label{old_est3} \| \Delta_{ij} \|_{H^m} \le C \epsilon e^{2\kappa_0 -2a\tau } \tilde{E}^{1/2}_m\,,\\
\label{old_est4} \| \Delta_{\psi} \|_{H^m} \le C \epsilon e^{-2a\tau} \tilde{E}^{1/2}_m\,,\\
\label{est1}\| \Delta_e  \|_{H^m} \le C \epsilon e^{-2a\tau} \tilde{E}^{1/2}_m\,,\\
\label{est2}\| \Delta_{b,i}  \|_{H^m} \le C \epsilon e^{p\tau+\kappa_0 -2a\tau } \tilde{E}^{1/2}_m\,,
\end{align}
hold on $I$, where the constant coefficients depend on $n$, $p$, $m$ and $c_1$. In addition, the estimates
\begin{align}
\label{old_est5} \| \Delta_{\mathcal{LS}}[\partial^\alpha \tilde{\mathbf{u}}]  \|_{L^1} \le C \epsilon e^{-a\tau} E_{\mathcal{LS},m}\,,\\
\label{old_est6} \| \Delta_{\mathcal{SH}}[\partial^\alpha u_j]  \|_{L^1} \le C \epsilon e^{-a\tau} E_{\mathcal{SH},m}\,,\\
\label{old_est7} \| \Delta_{\mathcal{M}}[\partial^\alpha h_{ij}]  \|_{L^1} \le C \epsilon e^{-a\tau} E_{\mathcal{M},m}\,,\\
\label{est3}\| \Delta_{\mathcal{SP}}[\partial^\alpha A_0]  \|_{L^1} \le C \epsilon e^{-a\tau} E_{\mathcal{SP},m}\,,\\
\label{est4}\| \Delta_{\mathcal{VP}}[\partial^\alpha A_i]  \|_{L^1} \le C \epsilon e^{-a\tau} E_{\mathcal{VP},m}\,,
\end{align}
hold on $I$, for $| \alpha | \le m$, where the constant coefficients depend here on $n$, $p$, $m$, $c_1$ and on $e^{-\kappa_1}$.
\end{lemma}

\begin{proof}
In Section 9.1 of \cite{Ringstrom08}, Ringstr\"om develops a systematic algorithm for estimating nonlinear terms similar to the $\Delta$-terms discussed here, presuming that an appropriate bootstrap assumption holds. He applies it in Section 6 of \cite{Ringstrom09} to a set of quantities very similar to those of interest here. The proof of the present lemma involves essentially the same application; we omit the details. We do, however, note that the $\Delta_{\mu \nu}$ quantities here differ from the corresponding quantities in \cite{Ringstrom09} only by terms which are quadratic in the electromagnetic fields. Since one determines that the electromagnetic fields have ``good scaling properties" according to the classification scheme of the algorithm of \cite{Ringstrom08}, the estimates calculated in \cite{Ringstrom09} for $\Delta_{\mu \nu}$ hold here.  Comparing the $\Delta_\psi$ quantity here and in \cite{Ringstrom09}, we find no differences, so the estimate from \cite{Ringstrom09} applies.  As for $\Delta_e$ and $\Delta_{b,i}$, the estimates \eqref{est1}-\eqref{est2} are again a straightforward application of the algorithm developed in \cite{Ringstrom08}. The same is true of the estimates \eqref{old_est5}-\eqref{est4}.

\end{proof}

We remark  here that Ringstr\"om's estimation algorithm works optimally with the energy functional $\tilde E_m$, with the scalings listed in \eqref{scalings} built into the definitions of $\tilde E_{\mathcal{LS},m}$, $\tilde E_{\mathcal{SH},m}$, etc. This is the prime reason for working with these rescalings.

The estimates just discussed, in Lemma \ref{Delta}, rely on the main bootstrap assumption, but do not depend on the fields satisfying any field equations. This next set of estimates do require that the field equations be satisfied, in addition to the main bootstrap assumption. 

\begin{lemma}
\label{Commut}
If the fields  $(g, \psi, A)$ satisfy  $\mathcal{MBA}_{\{p, a, c_1, \eta,  \tau_0, \kappa_1, n,  \epsilon\}}$ on an interval $I$, and also satisfy the field equations \eqref{refm1}-\eqref{refm6} on $I$, then the following estimates 
\begin{align}
\label{old_est8}\| [\tilde{\Box}_g, \partial^\alpha ] \tilde{\mathbf{u}}  \|_{L^2} \le C \epsilon e^{-2a\tau} \tilde{E}_m^{1/2}\,,\\
\label{old_est9}\| [\tilde{\Box}_g, \partial^\alpha ] u_i  \|_{L^2} \le C \epsilon e^{p\tau+\kappa_0 - 2a\tau } \tilde{E}_m^{1/2}\,,\\
\label{old_est10}\| [\tilde{\Box}_g, \partial^\alpha ] h_{ij}  \|_{L^2} \le C \epsilon e^{2\kappa_0 - 2a\tau } \tilde{E}_m^{1/2}\,,\\
\label{est5}\| [\tilde{\Box}_g, \partial^\alpha ]A_0  \|_{L^2} \le C \epsilon e^{-2a\tau} \tilde{E}_m^{1/2}\,,\\
\label{est6}\| [\tilde{\Box}_g, \partial^\alpha ]A_i  \|_{L^2} \le C \epsilon e^{p\tau+\kappa_0 -2a\tau} \tilde{E}_m^{1/2}\,,
\end{align}
hold on $I$, where $0 < |\alpha| \le m$ and where the constant coefficients depend on $n$, $p$, $m$, $c_1$ and on $e^{-\kappa_1}$.
\end{lemma}

\begin{proof}
Each of these commutators, when combined with the field equations \eqref{refm1}-\eqref{refm6}, produces a set of terms which includes one of the $\Delta$-terms along with terms which are linear in the fields (or the metric inverses). The metric inverses are easily handled using Lemma \ref{inverse}, while control of the $\Delta$-terms follows from Lemma \ref{Delta} above.
\end{proof}

With the estimates established in Lemmas \ref{Delta} and \ref{Commut} in hand, we may now derive the differential inequalities which control the growth of the energy functionals, and consequently control the growth of the fields and their derivatives. 

\begin{lemma}\label{diffin}
If the fields $(g, \psi, A)$ satisfy  $\mathcal{MBA}_{\{p, a, c_1, \eta,  \tau_0, \kappa_1, n,  \epsilon\}}$ as well as the field equations (\ref{refm1}) -- (\ref{refm6}) on an interval $I$, then the following differential inequalities 
\begin{align}
\label{ine1}\frac{d\tilde{E}_{\mathcal{LS},m}}{d\tau} &\le -2a \tilde{E}_{\mathcal{LS},m} + C \epsilon e^{-a\tau} \tilde{E}^{1/2}_{\mathcal{LS},m} \tilde{E}^{1/2}_m,\\
\label{ine2}\frac{d\tilde{E}_{\mathcal{SH},m}}{d\tau} &\le -2a \tilde{E}_{\mathcal{SH},m} + C\tilde{E}^{1/2}_{\mathcal{SH},m}(\tilde{E}^{1/2}_{\mathcal{LS},m} + \tilde{E}^{1/2}_{\mathcal{M},m}) + C \epsilon e^{-a\tau} \tilde{E}^{1/2}_{\mathcal{SH},m} \tilde{E}^{1/2}_m,\\
\label{ine3}\frac{d\tilde{E}_{\mathcal{M},m}}{d\tau} &\le C e^{-a\tau} \tilde{E}_{\mathcal{M},m} + C \tilde{E}^{1/2}_{\mathcal{LS}, m_0} \tilde{E}_{\mathcal{M},m} + C\tilde{E}^{1/2}_{\mathcal{LS},m} \tilde{E}^{1/2}_{\mathcal{M},m}+ C \epsilon e^{-a\tau} \tilde{E}^{1/2}_{\mathcal{M},m} \tilde{E}^{1/2}_m\,,\\
\label{ine4}\frac{d\tilde{E}_{\mathcal{SP},m}}{d\tau} &\le -2a \tilde{E}_{\mathcal{SP},m} + C \epsilon e^{-a\tau} \tilde{E}^{1/2}_{\mathcal{SP},m} \tilde{E}^{1/2}_m,\\
\label{ine5}\frac{d\tilde{E}_{\mathcal{VP},m}}{d\tau} &\le -2(p+3a) \tilde{E}_{\mathcal{VP},m}  + C \tilde{E}^{1/2}_{\mathcal{VP},m} \tilde{E}^{1/2}_{\mathcal{SP},m} + C \epsilon e^{-a\tau} \tilde{E}^{1/2}_{\mathcal{VP},m} \tilde{E}^{1/2}_m,
\end{align}
hold on $I$, where the constants depend on $n$, $p$, $m$, $c_1$ and on $e^{-\kappa_1}$.
\end{lemma}

\begin{proof}
Fundamentally, these differential inequalities are all obtained by calculating the time derivative of the energy functionals, applying the field equations, and then using the estimates of Lemmas \ref{Delta} and \ref{Commut} to control the non-linear terms. Since Corollary \ref{cor1} and \ref{cor2} (suitably applied) essentially contains the results of the time derivative calculations, the differential inequalities follow from that corollary combined with the appropriate estimates.  

We carry out the details here for $\tilde{E}^{1/2}_{\mathcal{SP},m}$ and $\tilde{E}^{1/2}_{\mathcal{VP},m}$, since the calculations for the others follow the same pattern, and differ little from those which are done in Section 7 of \cite{Ringstrom09}

Applying Corollary \ref{cor2} to the field equation (\ref{refm5}) for $A_0$, and then (i) using the definition of Sobolev norms, (ii) using the estimates from Lemmas \ref{Delta} and \ref{Commut}, and finally (iii) using the definition of $E_{\mathcal{SP},m}$, 
we obtain
\begin{align}
\frac{dE_{\mathcal{SP},m}}{d\tau} &\le - 2 \eta_\mathcal{SP} E_{\mathcal{SP},m}+ \sum_{| \alpha | \le m} \int_{\mathbf{T}^n} (\partial_\tau \partial^\alpha A_0 + \gamma_\mathcal{SP} \partial^\alpha A_0)(-\partial^\alpha \Delta_e + [\tilde{\Box}_g, \partial^\alpha ]A_0 )d\,x \nonumber\\
&\quad+ \sum_{| \alpha | \le m} \int_{\mathbf{T}^n}\Delta_{\mathcal{SP}}[\partial^\alpha A_0]  d\,x \nonumber\\ &\le - 2 \eta_\mathcal{SP} E_{\mathcal{SP},m}+ C \sum_{| \alpha | \le m} (\| \partial_\tau \partial^\alpha A_0 \|_2 + \|\partial^\alpha A_0 \|_2)(\|\partial^\alpha \Delta_e \|_2 + \| [\tilde{\Box}_g, \partial^\alpha ]A_0 \|_2) \nonumber\\
&\quad+ C\sum_{| \alpha | \le m} \| \Delta_{\mathcal{SP}}[\partial^\alpha A_0] \|_1 \nonumber\\
&\le - 2 \eta_\mathcal{SP} E_{\mathcal{SP},m}+ C ( \| \partial_\tau A_0 \|_{H^m} + \| A_0 \|_{H^m} ) (\| \Delta_e  \|_{H^m} + C \epsilon e^{-2a\tau} \tilde{E}^{1/2}_m ) \nonumber\\&\quad+ C \epsilon e^{-a\tau} E_{\mathcal{SP},m}\nonumber\\
&\le  - 2 \eta_\mathcal{SP} E_{\mathcal{SP},m}+ C e^{-a\tau} \tilde{E}^{1/2}_{\mathcal{SP},m} \epsilon e^{-2a\tau} \tilde{E}^{1/2}_m + C \epsilon e^{-a\tau} E_{\mathcal{SP},m}.
\end{align}
If we then substitute in the relation ${E}_{\mathcal{SP},m} := e^{-2a\tau} \tilde E_{\mathcal{SP}, m}$, we derive
\begin{equation}
- 2a e^{-2a\tau} \tilde{E}_{\mathcal{SP},m} +e^{-2a\tau} \frac{d\tilde{E}_{\mathcal{SP},m}}{d\tau} \le - 2 \eta_\mathcal{SP} E_{\mathcal{SP},m} + C \epsilon e^{-3a\tau} \tilde{E}^{1/2}_{\mathcal{SP},m}  \tilde{E}^{1/2}_m + C \epsilon e^{-a\tau} E_{\mathcal{SP},m};
\end{equation}
the differential inequality  (\ref{ine4}) then follows.

A similar string of calculations, starting from (\ref{refm6}), produces \eqref{ine5}. 
Specifically, we calculate
\begin{align}
\frac{dE_{\mathcal{VP},m}}{d\tau} &\le - 2 \eta_\mathcal{VP} E_{\mathcal{VP},m} +  \sum_i\sum_{| \alpha | \le m} \int_{\mathbf{T}^n}\Delta_{\mathcal{VP}}[\partial^\alpha A_i]  d\,x \nonumber\\
&\quad +\sum_i\sum_{| \alpha | \le m} \int_{\mathbf{T}^n} \partial_\tau \partial^\alpha A_i \left[-\partial^\alpha (2p e^{\tau+\tau_0} \partial_i A_0 +\Delta_{b,i}) + [\tilde{\Box}_g, \partial^\alpha ]A_i \right]d\,x\nonumber\\
&\le - 2\eta_\mathcal{VP} E_{\mathcal{VP},m} + C \epsilon e^{-a\tau} E_{\mathcal{VP},m}\nonumber\\
&\quad +\sum_i\sum_{| \alpha | \le m} \int_{\mathbf{T}^n} \partial_\tau \partial^\alpha A_i \left[-\partial^\alpha (2p e^{\tau+\tau_0} \partial_i A_0 +\Delta_{b,i}) + [\tilde{\Box}_g, \partial^\alpha ]A_i \right]d\,x\nonumber\\
&\le - 2\eta_\mathcal{VP} E_{\mathcal{VP},m} + C \epsilon e^{-a\tau} E_{\mathcal{VP},m} \nonumber\\
&\quad + \| \partial_\tau A_i \|_{H^m} (C \epsilon e^{p\tau + \kappa_0 -2a\tau} \tilde{E}^{1/2}_m + C e^{\tau+\tau_0}\|\partial_i A_0 \|_{H^m} + \| \Delta_{b,i}\|_{H^m})\nonumber\\
&\le  - 2\eta_\mathcal{VP} E_{\mathcal{VP},m} + C \epsilon e^{-a\tau} E_{\mathcal{VP},m}\nonumber\\
&\quad + C e^{2p\tau + 2\kappa_0 - 2a\tau} \tilde{E}^{1/2}_{\mathcal{VP},m} (\epsilon e^{-a\tau} \tilde{E}^{1/2}_m + \tilde{E}^{1/2}_{\mathcal{SP},m} ),
\end{align}
and then proceed to derive 
\begin{align}
e^{2p\tau +2\kappa_0 - 2a\tau} [ 2(p-a) \tilde{E}_{\mathcal{VP},m} + \frac{d\tilde{E}_{\mathcal{VP},m}}{d\tau} ]
\le -2\eta_\mathcal{VP} E_{\mathcal{VP},m} + C\epsilon e^{-a\tau} E_{\mathcal{VP},m} \nonumber\\
 + C \epsilon e^{2p\tau +2\kappa_0 - 3a\tau} \tilde{E}^{1/2}_{\mathcal{VP},m} \tilde{E}^{1/2}_m + Ce^{2p\tau +2\kappa_0 - 2a\tau} \tilde{E}^{1/2}_{\mathcal{VP},m} \tilde{E}^{1/2}_{\mathcal{SP},m}, \nonumber
\end{align}
from which  (\ref{ine5}) follows.
\end{proof}

%%%%%%%%%%%%%%%%%%%%%%%%%%%%%%%%%%%%%%%%%%%%%%%%%%%%%%%%%%%%%%%%%
\section{Global Existence}\label{S:global}

A key step in the proof of our main theorem is showing that if a set of initial data satisfies appropriate smallness conditions, then the development of that data exists for all future time. More specifically, as discussed in Section \ref{MainResults}, presuming that we have fixed a choice of the Einstein-Maxwell-scalar$_{\{n, V_0, \lambda\}}$ field theory, we consider an initial data set $(\Sigma^n, h, K, \varphi, \pi, B, E)$. If this data set, in a region $U\in \Sigma_n$, satisfies the smallness condition \eqref{epsilon} (which we justify below), then we calculate $\langle \varphi \rangle$  and $t_0=\Theta_{\{U,\zeta\}} (\Sigma^n, h, K, \varphi, \pi, B, E)$, thereby fixing a background solution $(\hat{g}, \hat{\phi}, \hat{A})_{ \{ t_0, p, c_0,\kappa \}}$ and a choice of time $t_0$. The choice of time fixes $\kappa_0 := \ln [4\ell(t_0)]$. We proceed to  smoothly glue $(\Sigma^n, h, K, \varphi, \pi, B, E)|_U$ into a copy of the background solution on $\mathbf{T}^n\setminus U$, producing data $( \tilde{h}, \tilde{K}, \tilde{\varphi}, \tilde{\pi}, \tilde{B}, \tilde{E})$ on $\mathbf{T}^n$ which agrees with the original data on $U$, satisfies the constraints everywhere except in a designated annulus, and satisfies the smallness condition  \eqref{epsilon} everywhere on $\mathbf{T}^n$. We now show that the development of this data exists for all future time.

The work of Section \ref{Reform} shows that to prove global existence for the data $(\Sigma^n, h, K, \varphi, \pi, B, E)|_U$ (or the data  $( \tilde{h}, \tilde{K}, \tilde{\varphi}, \tilde{\pi}, \tilde{B}, \tilde{E})$) evolved via equations \eqref{fe1}-\eqref{fe2}, it is sufficient to prove global existence for corresponding data, in the form $\{u, u_i, \gamma_{ij}, \psi, A_0, A_i\}$, evolved via equations (\ref{refm1}) - (\ref{refm6}). To carry this out, it is useful to first clarify the data correspondence:

\begin{definition}
For a fixed choice of ${\{n, V_0, \lambda\}}$---and therefore, correspondingly, a fixed choice of $p$ and $c_0$ (see \eqref{lambda} and \eqref{c0})---let  $( \tilde{h}, \tilde{K}, \tilde{\varphi}, \tilde{\pi}, \tilde{B}, \tilde{E})$ be a set of initial data on $\mathbf{T}^n$, with corresponding values of $\langle \varphi \rangle$, $t_0$, and $\kappa_0$. 
The \textbf{initial data for} (\ref{refm1}) - (\ref{refm6}) \textbf{associated with} $(\tilde{h}, \tilde{K}, \tilde{\varphi}, \tilde{\pi}, \tilde{E}, \tilde{B})$ at $\tau=0$ is given by
\begin{align}
\label{inid1}u(0, \cdot) = 0,\quad (\partial_\tau u)(0, \cdot) = 2np - 2t_0 (\mathrm{tr} \tilde{K}),\\
u_i(0, \cdot) = 0,\quad (\partial_\tau u_i)(0, \cdot) = \frac{1}{2} t_0 \tilde{h}^{kl} (2 \partial_k \tilde{h}_{li} - \partial_i \tilde{h}_{kl}),\\
\gamma_{ij}(0, \cdot) = \tilde{h}_{ij},\quad (\partial_\tau \gamma_{ij})(0, \cdot) = 2 t_0 \tilde{K}_{ij} - 2p \tilde{h}_{ij},\\
\psi(0, \cdot) = \tilde{\varphi} - \langle \tilde{\varphi} \rangle,\quad (\partial_\tau \psi)(0, \cdot) = t_0 \tilde{\pi} - \frac{2}{\lambda},\\
A_0(0, \cdot) = 0,\quad \partial_\tau A_0 (0, \cdot) = t_0 \tilde{h}^{ij}  \partial_i \tilde{A}_{j}(x),\\
\label{inid6}A_i(0, \cdot) = \tilde{A}_{i} (x),\quad \partial_\tau A_i (0, \cdot) = -t_0\tilde{E}_i,
\end{align}
where all the indices are defined with respect to the standard coordinates on $\mathbf{T}^n$.
\end{definition}

Our global existence theorem for this data is as follows: 
\begin{theorem}
\label{global}
For a fixed choice of ${\{n, V_0, \lambda\}}$, let $\eta$ and $a$ be given by (\ref{con1}) and (\ref{con2}) respectively and let $m_0 > n/2 +1$ be an integer. Let smooth data $(\tilde{h}, \tilde{K}, \tilde{\varphi}, \tilde{\pi}, \tilde{E}, \tilde{B})$ be specified  on $\mathbf{T}^n$, with corresponding values of $\langle \varphi \rangle$, $t_0$, and $\kappa_0$, and with associated data specified by (\ref{inid1}) - (\ref{inid6}).
Assume that there is a constant $c_1 > 2$ such that
\begin{equation}
\label{eq2}\frac{2}{c_1} | v |^2 \le e^{-2\kappa_0} \gamma_{ij}(0, x) v^i v^j \le \frac{c_1}{2} | v |^2
\end{equation}
holds for any $v \in \mathbf{R}^n$ and for $x \in \mathbf{T}^n$. There exist constants $\epsilon_0 >0$ and $c_b \in (0, \frac{2}{3})$ (depending only on $n$, $m_0$, $p$ and $c_1$) such that if
\begin{equation}
\tilde{E}^{1/2}_{m_0}(0) \le c_b \epsilon,
\end{equation}
holds for some $\epsilon \le \epsilon_0$, then there is a global solution to  (\ref{refm1}) - (\ref{refm6}). Furthermore, 
\begin{equation}
\label{eq3}\tilde{E}^{1/2}_{m_0}(\tau) \le \epsilon,
\end{equation}
holds for all $\tau \ge 0$, as do the primary bootstrap conditions (\ref{p0}) - (\ref{p2}).
\end{theorem}

\begin{proof}
We fix a choice of $\epsilon$ (possibly to be adjusted later), and we define 
$\mathcal{A}$ to be the set of $s \in [0, \infty)$ for which there exists a set of fields $(g,\phi, A)$ on the spacetime manifold $[0, s) \times \mathbf{T}^n$ such that i) the fields $(u, u_i, \gamma_{ij}, \psi, A_0, A_i)$ corresponding to $(g, \phi, A)$ satisfy the field equations (\ref{refm1}) - (\ref{refm6}) on $[0, s) \times \mathbf{T}^n$, and agree with the initial data (\ref{inid1}) - (\ref{inid6}) at $\tau=0$; and (ii), these fields satisfy the main bootstrap assumption $\mathcal{MBA}_{\{p, a, c_1, \eta,  \tau_0, \kappa_1, n,  \epsilon\}}$ on $[0, s) \times \mathbf{T}^n$. To prove the theorem, we need to show that $\mathcal{A}$ is non-empty, relatively open, and relatively closed (as a subset of $[0, \infty)).$\footnote{It is clear from its definition that $\mathcal{A}$ is connected.}

To show that  $\mathcal{A}$ is non-empty, we first note that, by construction, the system (\ref{refm1}) - (\ref{refm6}) is hyperbolic; hence, local existence for  solutions of the initial value problem follows.\footnote{We may also argue local existence for the initial value problem for (\ref{refm1}) - (\ref{refm6}) by invoking the equivalence of this system, up to gauge, to the Einstein-Maxwell-scalar field system (as argued in Section \ref{Reform}), and then relying on standard well-posedness  theorems for the Einstein-Maxwell-scalar field system.} In particular, there exists an interval $(\tau_{min}, \tau_{max})$ containing $\tau=0$ on which there exists a unique (smooth) solution to  (\ref{refm1}) - (\ref{refm6}). To verify that this solution satisfies the main bootstrap assumption for some nonvanishing neighborhood of $\tau=0$,  we observe that since the initial data satisfy the conditions $u(0, \cdot) = 0$, $u_i(0, \cdot) = 0$, and (\ref{eq2}), it follows that the primary bootstrap conditions (\ref{p0}) - (\ref{p2}) hold on some such neighborhood. Since the initial data also satisfies 
$\tilde{E}^{1/2}_{m_0}(0) \le \frac{2}{3} \epsilon$, we see that \eqref{eq3} holds on some neighborhood of $\tau=0$. We thus verify $\mathcal{MBA}_{\{p, a, c_1, \eta,  \tau_0, \kappa_1, n,  \epsilon\}}$, and consequently 
determine that $\mathcal{A}$ is non-empty.

To show that $\mathcal{A}$ is relatively closed, it is sufficient to show that if we assume that $T\in [0,\infty)$ is contained in $\mathcal{A}$, then the solution on $[0,T)$ extends past $T$, and $\mathcal{MBA}_{\{p, a, c_1, \eta,  \tau_0, \kappa_1, n,  \epsilon\}}$ holds at $\tau=T$.\footnote{This proves that $\mathcal{A}$ is relatively closed, since it shows that if $[0, T) \subset \mathcal{A}$, then $[0, T] \subset \mathcal{A}$.} It follows from the main bootstrap assumption $\mathcal{MBA}_{\{p, a, c_1, \eta,  \tau_0, \kappa_1, n,  \epsilon\}}$ and the field equations (\ref{refm1}) - (\ref{refm6}) that the fields $(u, u_i, \gamma_{ij}, \psi, A_0, A_i)$ defined on $[0, T)$ extend smoothly to $ T$, with the metric maintaining its Lorentzian character. The well-posedness of the field equations with initial data at $ T$ (as argued above) then guarantees that the solution extends past $ T$, and one readily verifies from the definition of $\mathcal{MBA}_{\{p, a, c_1, \eta,  \tau_0, \kappa_1, n,  \epsilon\}}$ that if the solution extends and if the main bootstrap assumption holds on $[0, T)$ then it also holds on $[0,T]$.

It remains to show that $\mathcal{A}$ is relatively open. To do this, it is sufficient to show that, in the situation just discussed, the main bootstrap assumption holds not just on $[0,T]$, but on some interval $[0,\hat T)$ with $\hat T>T.$ To show this, it is sufficient to show that \emph{an improved version} of $\mathcal{MBA}_{\{p, a, c_1, \eta,  \tau_0, \kappa_1, n,  \epsilon\}}$---i.e., a version with tighter estimates---holds on $[0,T)$. 

The improvements we obtain on the $\mathcal{PBA}_{\{p, a, c_1, \eta, \tau_0, \kappa_1, n\}}$ portion of the bootstrap assumption involve replacing the inequalities \eqref{p0}-\eqref{p2} in Definition \ref{PrimaryBootstrap} by strict inequalities. One does this by combining  inequality \eqref{energycontrol} from the main bootstrap assumption with Lemma \ref{control lemma} together with Sobolev embeddings to obtain control on the evolution of the metric components. The argument is identical to that in Theorem 4 of \cite{Ringstrom09}. 

The key to improving the energy estimate \eqref{energycontrol} is the set of differential inequalities in Lemma \ref{diffin} which control the evolution of the energy functionals. Since $\mathcal{MBA}_{\{p, a, c_1, \eta,  \tau_0, \kappa_1, n,  \epsilon\}}$ holds on the interval $[0,T)$, these inequalities hold there as well. We consider first \eqref{ine1}. Combining this with assumption \eqref{energycontrol}, we have
\begin{equation}
\frac{d\tilde{E}_{\mathcal{LS},m_0}}{d\tau} \le -2a \tilde{E}_{\mathcal{LS},m_0} + C \epsilon^2 e^{-a\tau} \tilde{E}^{1/2}_{\mathcal{LS},m_0}.
\end{equation}
Integrating this first-order ordinary differential inequality, we obtain
\begin{equation}
\label{intine1}
\tilde{E}^{1/2}_{\mathcal{LS},m_0} (\tau) \le e^{-a \tau} \tilde{E}^{1/2}_{\mathcal{LS},m_0} (0) + \frac{1}{2} C\tau e^{-a\tau} \epsilon ^2
\end{equation}
for all $\tau \in [0,T).$
Then using assumption \eqref{energycontrol} to estimate $\tilde{E}^{1/2}_{\mathcal{LS},m_0} (0)$, we get
\begin{equation}
\label{ine1+}
\tilde{E}^{1/2}_{\mathcal{LS}, m_0}(\tau) \le C_{\mathcal{LS}} (c_b \epsilon + \epsilon^2) e^{-a\tau/2}.
\end{equation}

Since the differential inequality \eqref{ine4}  for $\tilde{E}_{\mathcal{SP},m} (\tau)$ is essentially the same as that for $\tilde{E}_{\mathcal{LS},m_0} (\tau)$, we derive the same estimate for this quantity: 
\begin{equation}
\label{ine4+}
\tilde{E}^{1/2}_{\mathcal{SP}, m_0}(\tau) \le C_{\mathcal{SP}} (c_b \epsilon + \epsilon^2) e^{-a\tau/2}.
\end{equation}

To handle $\tilde{E}_{\mathcal{M}, m_0}(\tau)$, which satisfies \eqref{ine3}, we first use the estimate \eqref{ine1+} to simplify \eqref{ine3}, and then use the integrating factor $f=exp[\frac{C}{a}(e^{-a\tau}-1)]$ to solve the resulting first order ordinary differential inequality. We obtain, presuming that $c_b\le 1$  (see Theorem 4 of \cite{Ringstrom09}),  
\begin{equation}
\label{ine3+}
\tilde{E}_{\mathcal{M}, m_0} (\tau) \le  C_\mathcal{M} (c_b\epsilon + \epsilon^2)\epsilon.
\end{equation}

To control $\tilde{E}_{\mathcal{SH}, m_0} (\tau)$, which satisfies \eqref{ine2}, we substitute the estimates derived above for $\tilde{E}_{\mathcal{LS}, m_0} (\tau)$ and for $\tilde{E}_{\mathcal{M}, m_0} (\tau)$, together with the bootstrap estimate assumed for $\tilde{E}_{m_0}(\tau)$, into the differential inequality \eqref{ine2}, thereby obtaining
\begin{equation}
\frac{d\tilde{E}_{\mathcal{SH},m_0}}{d\tau} \le -2a \tilde{E}_{\mathcal{SH},m_0} + C_{\mathcal{SH}}(c_b^{1/2} \epsilon +\epsilon^{3/2})\tilde{E}^{1/2}_{\mathcal{SH},m_0}.
\end{equation}
We see that the right hand side of this inequality is negative if $\tilde{E}_{\mathcal{SH},m_0}$ exceeds a certain value; hence we obtain the upper bound 
\begin{equation}
\label{ine2+}
\tilde{E}^{1/2}_{\mathcal{SH}, m_0} (\tau) \le \frac{C_\mathcal{SH}}{2a} ( c_b^{1/2} \epsilon + \epsilon^{3/2})
\end{equation}
for this energy functional.

The remaining energy functional to control is $\tilde{E}_{\mathcal{VP},m_0}(\tau)$. Comparing the differential inequalities \eqref{ine3} and \eqref{ine5}, and noting from above (see estimates \eqref{ine1+} and \eqref{ine4+}) the identical estimates for $\tilde{E}^{1/2}_{\mathcal{LS}, m_0}(\tau)$ and for $\tilde{E}^{1/2}_{\mathcal{SP}, m_0}(\tau)$ we see that the differential inequality for $\tilde{E}_{\mathcal{VP},m_0}(\tau)$ is the same as that for $\tilde{E}_{\mathcal{M},m_0}(\tau)$, except that it includes a decay term, $-2(p+3a) \tilde{E}_{\mathcal{VP}, m_0}$, and that it does not include  terms similar to $Ce^{-a\tau} \tilde{E}_{\mathcal{M},m}$ and $C\tilde{E}^{1/2}_{\mathcal{LS},m_0}\tilde{E}_{\mathcal{M},m}$. It immediately follows that 
\begin{equation}
\label{ine5+}
\tilde{E}_{\mathcal{VP}, m_0} (\tau) \le  C_\mathcal{VP} (c_b\epsilon + \epsilon^2)\epsilon.
\end{equation}

We may now combine estimates \eqref{ine1+}, \eqref{ine4+}, \eqref{ine3+}, \eqref{ine2+}, and \eqref{ine5+}, together with sufficiently small choices of $c_b$ and $\epsilon$ (depending on the constants $C_\mathcal{LS}$, $C_\mathcal{M}$, $C_\mathcal{SH}$, $C_\mathcal{SP}$ and $C_\mathcal{VP}$) and conclude that
\begin{equation}
\tilde{E}^{1/2}_{m_0}(\tau) \le \frac{1}{3} \epsilon 
\end{equation}
holds in $[0, T)$. Thus we have an improvement on all estimates comprising  $\mathcal{MBA}_{\{p, a, c_1, \eta,  \tau_0, \kappa_1, n,  \epsilon\}}$. This proves global existence, thereby concluding the proof of this theorem.
\end{proof}

The global existence theorem just proven guarantees control of the energy functionals up to some chosen order $m_0$. In fact, we can show that the energy functionals to all orders $m$ are finite:
 
\begin{theorem}
\label{thmorderm}
For any solution of the field equations (\ref{refm1}) - (\ref{refm6}) which satisfies the hypothesis of Theorem \ref{global}, there exists a sequence of constants $C_m$  such that for every positive integer $m$, the corresponding energy functional satisfies 
\begin{equation}
\label{asym0}
\tilde{E}^{1/2}_m(\tau) \le C_m.
\end{equation}
 for all $\tau \ge 0$.
\end{theorem}

\begin{proof}
It follows from Theorem \ref{global} that the solutions under consideration satisfy the estimates of  $\mathcal{MBA}_{\{p, a, c_1, \eta,  \tau_0, \kappa_1, n,  \epsilon\}}$ for all time $\tau$. Hence the differential inequalities 
 (\ref{ine1}) - (\ref{ine5}) hold for all $\tau \ge 0$. 
 
It is useful to work with a set of three sequences of rescaled energy functionals; we set
\begin{equation}
\label{rescales}
\hat {E}_{\mathcal{SH},m} := e^{-a\tau/2} \tilde{E}_{\mathcal{SH},m},\, \hat{E}_{\mathcal{LS},m} := e^{a\tau/2} \tilde{E}_{\mathcal{LS},m},\, \hat{E}_{\mathcal{SP},m} := e^{a\tau/2} \tilde{E}_{\mathcal{SP},m}.
\end{equation} 
Based on (\ref{ine1}), (\ref{ine2}), and (\ref{ine4}), we find that the differential inequalities for these rescaled quantities take the form
\begin{align}
\frac{d \hat{E}_{\mathcal{LS},m}}{d\tau} &\le -a \hat{E}_{\mathcal{LS},m} + C\epsilon e^{-3a\tau/4} \hat{E}_{\mathcal{LS},m}^{1/2} \tilde{E}^{1/2}_m\,,\nonumber\\
\frac{d \hat{E}_{\mathcal{SP},m}}{d\tau} &\le -a \hat{E}_{\mathcal{SP},m} + C\epsilon e^{-3a\tau/4} \hat{E}_{\mathcal{SP},m}^{1/2} \tilde{E}^{1/2}_m\,,\nonumber\\
\frac{d \hat{E}_{\mathcal{SH},m}}{d\tau} &\le -2a \hat{E}_{\mathcal{SH},m} + C e^{-a\tau/4} \hat{E}_{\mathcal{SH},m}^{1/2} ( \tilde{E}^{1/2}_{\mathcal{LS},m} + \tilde{E}^{1/2}_{\mathcal{M},m})\nonumber\\ 
&\quad\, + C\epsilon e^{-5a\tau/4} \hat{E}_{\mathcal{SH},m}^{1/2} \tilde{E}^{1/2}_m\,.\nonumber
\end{align}
If we now define the quantity
\begin{equation}
\label{defE}
\mathcal{E}_m := \hat{E}_{\mathcal{LS},m} + \hat{E}_{\mathcal{SH},m} + \tilde{E}_{\mathcal{M},m} +\hat{E}_{\mathcal{SP},m} + \tilde{E}_{\mathcal{VP},m}
\end{equation}
then it follows from the above differential inequalities combined with 
(\ref{ine3}) and \eqref{ine5} that $\mathcal{E}_m(\tau)$ satisfies 
\begin{equation}
\label{eq4}
\frac{d \mathcal{E}_m}{d\tau} \le C e^{-a\tau/4}\mathcal{E}_m + C \tilde{E}^{1/2}_{\mathcal{LS},m_0} \tilde{E}_{\mathcal{M},m}.
\end{equation}

We now set $m=m_0$. Noting the boundedness of $\tilde{E}_{\mathcal{M}, m_0} (\tau)$ (see \eqref{eq3}) and noting the behavior of $\tilde{E}^{1/2}_{\mathcal{LS},m_0} (\tau)$ (which we infer from the relationship \eqref{rescales} between $\tilde{E}^{1/2}_{\mathcal{LS},m_0} (\tau)$ and $\hat{E}^{1/2}_{\mathcal{LS},m_0} (\tau)$, and the relationship \eqref{defE} between $\hat{E}^{1/2}_{\mathcal{LS},m_0} (\tau)$ and $\mathcal{E}_{m_0})$, we derive 
\begin{equation}
 \frac{d \mathcal{E}_{m_0}}{d\tau} \le C e^{-a\tau/4}\mathcal{E}_{m_0}.
 \end{equation}
It follows from this differential inequality that 
 $\mathcal{E}_{m_0}$ is bounded\footnote{ We note that the proof of Theorem \ref{global} contains estimates of $\mathcal{E}_{m_0}$ and of $\tilde{E}_{\mathcal{LS},m_0} $ on a finite time  interval. Here, we show that the estimates hold for all $\tau \ge 0$.} and consequently that $\tilde{E}^{1/2}_{\mathcal{LS},m_0} \le C e^{-a\tau/4}$.

If we now substitute this estimate for $\tilde{E}^{1/2}_{\mathcal{LS},m_0}$ into \eqref{eq4}, we obtain (for arbitrary $m$)
\begin{equation}
\frac{d \mathcal{E}_m}{d\tau} \le C e^{-a\tau/4}\mathcal{E}_m.
\end{equation} 
It follows that  $\mathcal{E}_m$ is bounded for all $m$,  which implies that  $\tilde{E}_{\mathcal{LS},m}$, $\tilde{E}_{\mathcal{M},m}$, $\tilde{E}_{\mathcal{SP},m}$ and $\tilde{E}_{\mathcal{VP},m}$ are bounded as well. 

It remains to show that $\tilde{E}_{\mathcal{SH},m}$ is bounded. Substituting into (\ref{ine2})  the boundedness conditions just determined, we have 
\begin{equation}
\label{this}
 \frac{d \tilde{E}_{\mathcal{SH},m} }{d\tau} \le -2a \tilde{E}_{\mathcal{SH},m} + C e^{-a\tau} \tilde{E}_{\mathcal{SH},m} + C\tilde{E}_{\mathcal{SH},m}^{1/2}.
 \end{equation}
For $\tau$  large enough, the second term can be absorbed into the first term. It follows that $\tilde{E}_{\mathcal{SH},k}$ is bounded, since \eqref{this} requires that it decay if it exceeds certain value. This proves the theorem.
\end{proof}

%%%%%%%%%%%%%%%%%%%%%%%%%%%%%%%%%%%%%%%%%%%%%%%%%%%%%
\section{Causal Geodesic Completeness}\label{S:geodesic}

With global existence for spacetime developments of certain types of initial data sets proven above, we  argue in this section that causal geodesics with initial points contained in certain regions of those data sets are complete. We do this in two steps, with the first  proposition establishing certain geometric properties of these developments.

We note that since the behavior of causal geodesics depends exclusively on the metric, the discussion here is very similar to that in \cite{Ringstrom09}.

\begin{proposition}\label{geo1}
Let $(g, \phi, A )$ denote a spacetime development (satisfying the field equations (\ref{refm1}) - (\ref{refm6})) of the sort constructed in Theorem \ref{global}. Let $\gamma$ be a future-directed causal curve contained in this development, with parametrization domain $[s_0, s_{max})$, and  with $\gamma^0(s_0) = t_0$ ( $t_0$ is given by (\ref{t0})). If the $\epsilon$ parametrizing the smallness of the data (as stated in  Theorem \ref{global}) is sufficiently small (depending only on $n$, $p$ and $c_1$), then $\dot{\gamma}^0 > 0$ (for all $s$), and the projected spatial length of the path satisfies the condition 
\begin{equation}\label{eq15}
\int_{s_0}^{s_{max}} [ g_{ij}(t_0, \gamma_\flat) \dot{\gamma}^i \dot{\gamma}^j ]^{1/2} ds \le d(\epsilon) \ell(t_0),
\end{equation}
 where $d(\epsilon)$ is independent of $\gamma$, with $d(\epsilon) \to 1$ as $\epsilon \to 0$, and where $\gamma_\flat(s) = (\gamma^1(s), \gamma^2(s),\cdots, \gamma^n(s))$. Furthermore, assuming that $\gamma$ is future inextendible, we have $\gamma^0(s) \to \infty$ as $s \to s_{max}$.
\end{proposition}

\begin{proof}
Since $\gamma$ is a future-directed causal curve, we have the following inequalities
\begin{align}
\label{eq5}
g_{\mu\nu} \dot{\gamma}^\mu \dot{\gamma}^\nu \le 0,\\
\label{eq13}
g_{00} \dot{\gamma}^0 + g_{0i} \dot{\gamma}^i < 0.
\end{align}
It then follows from (\ref{p2}) (one of the estimates comprising the primary bootstrap assumption, which we have shown in Theorem \ref{global} holds for $(g, \phi, A )$) that for some constant $\eta \in (0,1)$,
\begin{equation}
| 2 g_{0i} \dot{\gamma}^0 \dot{\gamma}^i | \le \eta^{1/2} | \dot{\gamma}^0 |^2 + \eta^{-1/2} | g_{0i} \dot{\gamma}^i |^2 \le \eta^{1/2} | \dot{\gamma}^0 |^2 + \eta^{1/2}  c_1^{-1} e^{2p\tau +2\kappa_0 - 2a\tau} \delta_{ij} \dot{\gamma}^i \dot{\gamma}^j.
\end{equation}
As a consequence of  (\ref{p0}) (another of the proven primary bootstrap estimates), the last term is bounded by $\eta^{1/2} g_{ij} \dot{\gamma}^i \dot{\gamma}^j$. Combining  (\ref{p1}) and (\ref{eq5}), we find that
\begin{equation}\label{eq12}
g_{ij} \dot{\gamma}^i \dot{\gamma}^j \le c(\eta) \dot{\gamma}^0 \dot{\gamma}^0,
\end{equation}
where $c(\eta) \to 1$ as $\eta \to 0+$. Again invoking (\ref{p0}),  we determine that
\begin{equation}\label{eq6}
\delta_{ij} \dot{\gamma}^i \dot{\gamma}^j \le c_1 c(\eta) e^{-2p\tau -2\kappa_0} \dot{\gamma}^0 \dot{\gamma}^0 = c_1 c(\eta)(t/t_0)^{-2p}e^{-2\kappa_0} \dot{\gamma}^0 \dot{\gamma}^0.
\end{equation}

Now, as a consequence of the energy estimate (\ref{eq3}), the role of the energy functionals in controlling function norms (see \eqref{es0}), and Sobolev embedding (recall the condition in Theorem \ref{global} that $m_0 > n/2 + 1$), we have 
\begin{equation}
e^{-2\kappa_0 + a\tau} \| \partial_\tau \gamma_{ij} \|_\infty \le C\epsilon.
\end{equation}
Consequently, we  obtain 
\begin{equation}
 \| (t/t_0)^{-2p}e^{-2\kappa_0} g_{ij}(t,\cdot) - e^{-2\kappa_0} g_{ij}(t_0,\cdot) \|_\infty \le C a^{-1} \epsilon,
 \end{equation}
where $C$ depends only on $n$, $p$ and $c_1$. Combining this with (\ref{eq6}), we obtain
\begin{equation}
\label{that}
 |e^{-2\kappa_0} g_{ij}(t_0,\gamma_\flat) \dot{\gamma}^i\dot{\gamma}^j - (t/t_0)^{-2p}e^{-2\kappa_0}g_{ij} \dot{\gamma}^i\dot{\gamma}^j | \le C a^{-1} \epsilon c_1 c(\eta) (t/t_0)^{-2p}e^{-2\kappa_0} \dot{\gamma}^0 \dot{\gamma}^0.
 \end{equation}
We note here that it follows from  (\ref{eq11}), (\ref{eq10}) and (\ref{eq3}) that\footnote{The argument is as follows: From (\ref{eq11}), (\ref{eq10}) and (\ref{eq3}), we have $| u | \le \| u \|_{H^m} \le C e^{-a\tau} \tilde{E}^{1/2}_{LS, m} \le C\epsilon e^{-a\tau}$ and we have $| u_j |^2 \le \| u_j \|^2_{H^m} \le C^2 e^{2p\tau+2\kappa_0-2a\tau} \tilde{E}_{SH,m} \le C^2 \epsilon^2 e^{2p\tau+2\kappa_0-2a\tau}$. If $\epsilon$ is small enough, we can replace $\eta$ by $\epsilon$.}
$\eta$ in (\ref{p1}) and (\ref{p2}) can be replaced by $C\epsilon$, where $C$ only depends on $n$, $p$ and $c_1$. Combining this observation with the two inequalities \eqref{eq12}  and \eqref{that} we produce the inequality
\begin{equation} 
\label{eq14}
e^{-2\kappa_0} g_{ij}(t_0,\gamma_\flat) \dot{\gamma}^i\dot{\gamma}^j \le d^2(\epsilon)  (t/t_0)^{-2p}e^{-2\kappa_0} \dot{\gamma}^0 \dot{\gamma}^0,
\end{equation}
where $d(\epsilon) \to 1$ as $\epsilon \to 0+$. 

To show that $\dot{\gamma}^0 > 0$, we consider (\ref{eq13}). Combining (\ref{p2}) and (\ref{eq6}) we derive 
\begin{equation}
\label{theother}
|g_{0i} \dot{\gamma}^i | \le [ e^{-2p\tau -2\kappa_0} \delta^{ij} g_{0i} g_{0j} ]^{1/2} [ e^{2p\tau +2\kappa_0} \delta_{ij} \dot{\gamma}^i \dot{\gamma}^j ]^{1/2} \le \xi(\epsilon) | \dot{\gamma}^0 |,
\end{equation}
where $\xi(\epsilon) \to 0$ as $\epsilon \to 0+$. If we now make $\epsilon$ sufficiently small (depending only on $n$, $p$ and $c_1$), it follows  from \eqref{theother} and \eqref{eq13} that $\dot{\gamma}^0 > 0$.
Combining  this result with (\ref{eq14}), we obtain the estimate (\ref{eq15}). 

To complete the proof, we assume 
that $\gamma$ is future inextendible and we suppose that $\gamma^0$ does not tend to $\infty$. Since $\dot{\gamma}^0 >0$, it follows that $\gamma^0$ has to converge to a finite number. Examining (\ref{eq6}), we see that $\gamma^i$ must converge as well; hence $\gamma_\flat$ must converge to a point on $\mathbf{T}^n$. But this contradicts the assumption that $\gamma$ is a future  inextendible causal path.
\end{proof}

We proceed now to prove causal geodesic completeness:

\begin{proposition}
\label{geo2}
For sufficiently small $\epsilon$ (depending on $n, p$, and $c_1$), the spacetime development (satisfying the field equations (\ref{refm1}) - (\ref{refm6})) of a set of initial data satisfying the hypothesis of 
Theorem \ref{global} is future causally geodesically complete.
\end{proposition}

\begin{proof}
We let $\gamma$ be a future-directed causal geodesic and we let  $(s_{min}, s_{max})$ denote the maximum range of its (proper) parameter. The geodesic equation for $\gamma$ takes the form
\begin{equation}
\ddot{\gamma}^{\beta} + \Gamma^{\beta}_{\mu \nu} \dot{\gamma}^\mu\dot{\gamma}^\nu = 0; 
\end{equation}
for the purposes of our argument here, we work with the $\beta=0$ component of this equation:
\begin{equation}
\label{eq16}
\ddot{\gamma}^0+ \Gamma^0_{\mu \nu} \dot{\gamma}^\mu\dot{\gamma}^\nu = 0.
\end{equation}

We wish to estimate the second term of this equation. First, using the estimation algorithm Section 9.1 of \cite{Ringstrom08} together with the energy estimate \eqref{eq3}, we obtain
\begin{align}
\nonumber | \Gamma^0_{00}| &\le C \epsilon (p/t) e^{-a\tau}, \\
\nonumber | \Gamma^{0}_{0i} | &\le C \epsilon (p/t) e^{p\tau +\kappa_0-a\tau}, \\
\nonumber | \Gamma^{0}_{ij}- (p/t) g_{ij} | &\le C \epsilon (p/t) e^{2p\tau +2\kappa_0 -a\tau}.
\end{align}
Consequently,  for $t$ large enough or $\epsilon$ small enough, we have $\Gamma^{0}_{ij}\dot{\gamma}^i \dot{\gamma}^j \ge 0$. Combining these estimates with (\ref{eq6}), we conclude that
\begin{equation}
 | \Gamma^0_{00} \dot{\gamma}^0 \dot{\gamma}^0 | + 2 |  \Gamma^{0}_{0i} \dot{\gamma}^0 \dot{\gamma}^i | \le C \epsilon (p/t) e^{-a\tau} | \dot{\gamma}^0 |^2, 
 \end{equation}
where $C$  depends only on $n$, $p$ and $c_1$. Based on these conclusions together with (\ref{eq16}), we find that
\begin{equation}
\label{gamma..}
 \ddot{\gamma}^0 \le C\epsilon (p/t) e^{-a\tau} \dot{\gamma}^0 \dot{\gamma}^0 = C \epsilon (p/t) (t/t_0)^{-a} \dot{\gamma}^0 \dot{\gamma}^0, 
 \end{equation}
for $s \ge s_1$. We note here that $\gamma^0$ and $t$ can be used interchangeably.

Since $\dot{\gamma}^0 > 0$ (assuming $\epsilon$ to be small enough),
we may divide both sides of \eqref{gamma..} by $\dot{\gamma}^0$ and then integrate, thereby producing 
\begin{equation} 
\ln \frac{\dot{\gamma}^0(s)}{\dot{\gamma}^0(s_1)} \le C \epsilon p \int_{s_1}^s t^{-1} (t/t_0)^{-a} \dot{\gamma}^0 ds = C \epsilon p \int_{\gamma^0(s_1)}^{\gamma^0(s)} t^{-1} (t/t_0)^{-a} dt \le C \epsilon p /a, 
\end{equation}
where we let $s_1$ be large enough such that $\gamma^0(s_1) \ge t_0$ if necessary. It follows that $\dot{\gamma}^0$ is bounded away from 0. Hence we have
\[ \gamma^0(s) - \gamma^0(s_0) = \int_{s_0}^s \dot{\gamma}^0(s) ds \le C | s - s_0|, \]
for some constant $C$. Since $\gamma^0(s) \to \infty$ as $ s \to s_{max}$, we conclude that $s_{max} = \infty$. Thus $\gamma$ is future complete.
\end{proof}

%%%%%%%%%%%%%%	5555555%%%%%%%%%%%%%%%%%%%%%%%%%%%%%%%%%%%%%
\section{Asymptotic Expansions}\label{S:asymptotic}

The results of Sections \ref{S:global} and \ref{S:geodesic} provide the bulk of the proof of  our main result, Theorem \ref{maintheorem} (we complete the proof below, in Section \ref{S:proof}). These results say little about asymptotic behavior of the solutions of interest, beyond global existence and geodesic completeness. Here, we address the issue of asymptotic behavior, as a step toward the proof of Theorem \ref{2ndtheorem}.

\begin{proposition}
\label{asym}
For sufficiently small $\epsilon$ (depending on $n,p, c_1$, and $m_0$),  the spacetime development $(g, \phi, A)$ (satisfying the field equations (\ref{refm1}) - (\ref{refm6})) of a set of initial data satisfying the hypothesis of Theorem \ref{global} has  the following asymptotic behavior:  There exists a smooth Riemannian metric $H$ on $\mathbf{T}^n$, and for every integer $l \ge 0$ there exists a constant $\alpha_l$ (depending only on $n$, $p$, $c_1$ and $l$), such that for all $t \ge t_0$, the following estimates hold:
\begin{align}
\label{asym1} \| \phi(t,\cdot) - \frac{2}{\lambda} \ln t + \frac{c_0}{\lambda} \|_{C^l} + \| (t\partial_t\phi)(t,\cdot) - \frac{2}{\lambda} \|_{C^l} \le \alpha_l (t/t_0)^{-a},\\
\label{asym2} \| E_i \|_{C^l} = \| \partial_i A_0 - \partial_0 A_i \|_{C^l} \le \alpha_l e^{\kappa_0} (t/t_0)^p \,(t/t_0)^{-1-a},\\
\label{asym3} \| B_{ij} \|_{C^l} = \| \partial_i A_j - \partial_j A_i \|_{C^l} \le \alpha_l e^{2\kappa_0} (t/t_0)^{2p} \,(t/t_0)^{-1-a},\\
\label{asym4} \| (1 + g_{00})(t, \cdot) \|_{C^l} + \| (t \partial_t g_{00})(t, \cdot) \|_{C^l} \le \alpha_l (t/t_0)^{-a},\\
\label{asym5} \| \frac{1}{t}g_{0i}(t, \cdot) - \frac{1}{(n-2)p +1}H^{jm} \gamma_{jim} \|_{C^l} + \| t\partial_t (\frac{1}{t}g_{0i})(t, \cdot) \|_{C^l} \nonumber\\ \le \alpha_l (t/t_0)^{-a}\,,\\
\label{asym6}\| e^{-2\kappa_0} (t/t_0)^{-2p} g_{ij}(t, \cdot) - H_{ij} \|_{C^l} \nonumber\\ +  \| e^{-2\kappa_0} (t/t_0)^{-2p} (t\partial_t g_{ij})(t, \cdot) - 2pH_{ij}\|_{C^l} \le \alpha_l (t/t_0)^{-a}\,,\\
\label{asym7}\|  e^{2\kappa_0}(t/t_0)^{2p} g^{ij}(t, \cdot) - H^{ij} \|_{C^l} \le \alpha_l (t/t_0)^{-a}\,,\\
\label{asym8}\| e^{-2\kappa_0} (t/t_0)^{-2p} t\,K_{ij}(t, \cdot) - pH_{ij} \|_{C^l}\le \alpha_l (t/t_0)^{-a}\,.
\end{align}
Here $\gamma_{jim}$ denote the (lowered index) Christoffel symbols associated to the metric $H$, $K_{ij}(t, \cdot)$ are the components of the second fundamental form induced on the Cauchy surface $\{t\}\times \mathbf{T}^n$ by the spacetime metric $g_{\mu\nu}$, 
and  $\| \cdot \|_{C^l}$ denotes the $C^l$ norm on $\mathbf{T}^n$.
\end{proposition}

\begin{remark}
As noted above in Section \ref{MainResults}, although it may appear from \eqref{asym2} and \eqref{asym3}  that the electromagnetic field grows exponentially, this is essentially a coordinate effect, which reflects the use of coordinate bases relative to which the physical metric is expanding. If one examines the locally measured
fields (factoring out the metric expansion) then the electric field and magnetic field decay as $(t/t_0)^{-(1+a)}$. In fact, one expects that a further analysis would show a stronger decay rate.
\end{remark}

\begin{proof}
The first four of these estimates, (\ref{asym1}) - (\ref{asym4}), follow immediately from the energy bound in Theorem \ref{thmorderm} together with the norm bounds  (\ref{51}), (\ref{eq11}), (\ref{52}) and (\ref{53}).

Based on the norm estimate (\ref{es0}), on the relation  $\gamma_{ij} := (t/t_0)^{-2p} g_{ij} = e^{-2p \tau} g_{ij}$, and on Sobolev embedding, we see that 
the $C_l$ norm of $e^{-2\kappa_0} \partial_\tau \gamma_{ij}$ decays as $e^{-a\tau}$. Thus there exist smooth functions $H_{ij}$ such that for every integer $l \ge 0$,
\begin{equation}
\label{extra} 
\| e^{-2\kappa_0} \gamma_{ij}(\tau, \cdot) - H_{ij} \|_{C_l} \le \alpha_l e^{-a\tau} 
\end{equation}
holds for some $\alpha_l$ and for all $\tau > 0$. Combining this decay result with that for $\gamma_{ij}$ noted above, we obtain (\ref{asym6}). 

To derive \eqref{asym7}, we calculate 
\begin{equation}
e^{2\kappa_0} \partial_\tau (e^{2p\tau} g^{ij}) = 2p e^{2p\tau + 2\kappa_0} g^{ij} - e^{2p\tau + 2\kappa_0} g^{i\mu} g^{j\nu} \partial_\tau g_{\mu\nu}. 
\end{equation}
It follows from (\ref{es0}) that the  $C_l$ norm of the right hand side of the above decays as $e^{-a\tau}$. Hence there exist smooth functions $H^{ij}$ on $\mathbf{T}^n$ such that (\ref{asym7}) holds for all $\tau > 0$. We also conclude from above that $H^{ij} H_{jk} = \delta^i_k$, which implies that $H_{ij}$ is a Riemannian metric on $\mathbf{T}^n$ and $H^{ij}$ is its inverse.

To obtain the remaining estimates of this proposition, (\ref{asym5}) and (\ref{asym8}), we refer to the derivation carried out in Section 10 of \cite{Ringstrom09}. The corresponding analysis for the Einstein-Maxwell-scalar field theory requires that we handle an extra term: $F_{0\sigma} F_i\,^\sigma - \frac{1}{2(n-1)} g_{0i} F_{\alpha\beta} F^{\alpha\beta}$. We readily verify that this term decays sufficiently quickly for the argument of Section 10 of \cite{Ringstrom09} to apply here.
\end{proof}

%%%%%%%%%%%%%%%%%%%%%%%%%%%%%%%%%%%%%%%%%%%%%%%%%%%%%%%%%%%%%%%%%
\section{Proof of the Main Theorem}\label{S:proof}

The theorems and propositions of Sections \ref{S:global}, \ref{S:geodesic}, and \ref{S:asymptotic} essentially prove our main results, Theorem \ref{maintheorem} and Theorem \ref{2ndtheorem}. Here we discuss a few of the remaining details. We note that the steps needed here to complete these theorems are similar to those carried out in \cite{Ringstrom08} and \cite{Ringstrom09}; hence, we leave out some of the details.

\vspace{4mm}
{\em Step 1: Construction of a global-in-time patch, related to data on $U$.} Theorem \ref{global} of Section \ref{S:global} proves global-in-time existence for smooth initial data $(\tilde{h}, \tilde{K}, \tilde{\varphi}, \tilde{\pi}, \tilde{E}, \tilde{B})$ specified  on $\mathbf{T}^n$ and satisfying suitable smallness conditions. As noted in Section \ref{S:global}, a key step in using Theorem \ref{global} to prove our main result, Theorem \ref{maintheorem}, is to construct this smooth data on $\mathbf{T}^n$ from the given data  $(h, K, \varphi, \pi, E, B)$ specified  on $U\in \Sigma^n$ and satisfying the hypotheses of Theorem \ref{maintheorem}.  To do this, we define 
 a cutoff function $f_c(x) \in C_0^\infty (B_1(0))$ such that $f_c(x) = 1$ for $| x | \le 15/16$ and $0 \le f_c(x) \le 1$, and we define $(\tilde{h}, \tilde{K}, \tilde{\varphi}, \tilde{\pi}, \tilde{E}, \tilde{B})$ on $\mathbf{T}^n$ as follows:
\begin{align}
\tilde{h}_{ij} &= f_c (h_{ij}\circ x^{-1}) + (1 - f_c) e^{2\kappa_0} \delta_{ij},\\
\tilde{K}_{ij} &= f_c (K_{ij}\circ x^{-1}) + (1 - f_c) \frac{p}{t_0} e^{2\kappa_0} \delta_{ij},\\
\label{20}\tilde{\varphi} &= f_c \varphi\circ x^{-1} + (1 - f_c) \langle \varphi \rangle \nonumber\\ &\quad- \frac{1 - f_c}{1 - \langle f_c \rangle} [ \langle f_c (\varphi \circ x^{-1}) \rangle - \langle f_c \rangle \langle \varphi \rangle ],\\
\tilde{\pi} &= f_c (\pi \circ x^{-1}) + (1 - f_c) \frac{2}{\lambda t_0},\\
\label{sym:etilde}\tilde{E} &= f_c (E\circ x^{-1}),\\
\label{sym:btilde}\tilde{B} &= f_c (B\circ x^{-1}).
\end{align}
Here $t_0$ and $\kappa_0$ are defined in the statement of Theorem \ref{maintheorem}, and  the last term in (\ref{20}) is included so as to ensure that $t_0$ defined in Theorem \ref{global} equals that defined in Theorem \ref{maintheorem}. 

To go from the initial data set  $(\tilde{h}, \tilde{K}, \tilde{\varphi}, \tilde{\pi}, \tilde{E}, \tilde{B})$ just defined to the associated initial data for the evolution equations  (\ref{refm1}) - (\ref{refm6}), we rely on (\ref{inid1}) - (\ref{inid6}). Included in this transition is passage from the magnetic field $\tilde {B}(x)$ to the vector potential 
 $\tilde{A}_i(x)$. To show that this can always be done, with the needed norm control, we note the following:\footnote{This lemma is very similar to  Lemma 2.6 in \cite{Svedberg}. The main difference is that our version is specifically carried out for $\mathbf{T}^n$.}
\begin{lemma}
\label{lemmaAtilde}
Suppose that $\tilde{B}(x)$  is a closed 2-form on $\mathbf{T}^n$ such that $\mbox{supp}\,\tilde{B} \subset B_1(0)$ and $\tilde{B}_{ij} \in H^{m}(\mathbf{T}^n)$ for $m > n/2 + 1$. There exists a 1-form $\tilde{A}(x)$ on $\mathbf{T}^n$ such that (i) $\tilde{B} = d \tilde{A}$; and (ii) for any given $\epsilon > 0$, there exists $\delta > 0$ such that if
\begin{align}
\sum_{i,j} \| \tilde{B}_{ij} \|_{H^{m}(\mathbf{T}^n)} \le \delta,
\end{align}
then
\begin{align}
\sum_i \| \tilde{A}_i \|_{H^{m+1}(\mathbf{T}^n)} \le \epsilon.
\end{align}
\end{lemma}
\begin{proof}
Applying the Poincar\'e Lemma in the form of Problem 15-2 from  \cite{JackLee} to the closed two-form $\tilde{B}(x)$ on the ball $B_1(0)$, we obtain a 1-form $\bar{A}$ on $B_1(0)$ such that $\tilde{B} = d \bar{A}$ and $\bar{A}_i = \int_0^1 tx^j \tilde{B}_{ji}(tx)dt$ on $B_1(0)$. Recalling that $\tilde{B}_{ji}$ vanishes at the boundary of the ball $B_1(0)$, we see that $\bar{A}$ does as well, so we may extend $\bar{A}$ continuously to all of $\mathbf{T}^n$, with it vanishing on $\mathbf{T}^n \setminus B_1(0)$.

To control the norm of the vector potential, we now carry out a gauge transformation: We replace $\bar{A}$ by 
$\tilde{A} := \bar{A} + d\phi$, with $\phi$ taken to be a solution of the Poisson equation $\Delta \phi := - \sum_i \partial_i \bar{A}_i$. It follows from a straightforward calculation that the components $\tilde{A}_i$ of the gauge-transformed vector potential satisfy a Poisson equation with the source term being the divergence of $\tilde{B}$. Applying standard elliptic theory (e.g., Theorem 3 on page 249 of \cite{McOwen})
to this equation (on the closed manifold $\mathbf{T}^n$), we obtain estimates which lead to the conclusion of this lemma (see the proof of Lemma 2.6 in \cite{Svedberg} for details of these estimates, done on a bounded domain rather than on a closed manifold).
\end{proof}

The hypothesis of Theorem \ref{global} includes the condition that $\tilde{E}_{m_0}^{1/2}(0)$ be small. To show that this holds for the data  $(\tilde{h}, \tilde{K}, \tilde{\varphi}, \tilde{\pi}, \tilde{E}, \tilde{B})$ on $\mathbf{T}^n$  constructed as above (from data satisfying the hypothesis of the Theorem \ref{maintheorem}), we need to verify that for any  given any $\delta > 0$, there exists $\epsilon > 0$ such that if  (\ref{epsilon}) holds, then
$\tilde{E}_{m_0}^{1/2}(0) \le \delta.$
The arguments leading to the estimates $\tilde{E}^{1/2}_{\mathcal{LS}, m_0}(0) \le C\epsilon$, $\tilde{E}^{1/2}_{\mathcal{SH}, m_0}(0) \le C\epsilon$, and $\tilde{E}^{1/2}_{\mathcal{M}, m_0}(0) \le C\epsilon$ are very similar to those in Section 11 of \cite{Ringstrom09}. Here we show that $\tilde{E}^{1/2}_{\mathcal{SP}, m_0}(0)$ and $\tilde{E}^{1/2}_{\mathcal{VP}, m_0}(0)$ must  be arbitrarily small, presuming that we choose a small enough $\epsilon$.

To show that $\tilde{E}^{1/2}_{\mathcal{SP}, m_0}(0)$ can be made arbitrarily small, since $A_0(0,\cdot) = 0$, we only need to show that $ \| \partial_\tau A_0(0,\cdot) \|_{H^{m_0}} = \|  t_0 \tilde{h}^{ij} \partial_i \tilde{A}_{j}(x) \|_{H^{m_0}} $ can be made arbitrarily small. It is therefore sufficient  to show that $ \| \tilde{A}_i(x) \|_{H^{m_0+1}} $ can be made small enough.
To estimate $\tilde{E}^{1/2}_{\mathcal{VP}, m_0}(0)$, we need to estimate $\| e^{-\kappa_0} \partial_\tau A_i(0, \cdot) \|_{H^{m_0}}$ and $\| e^{-\kappa_0}  A_i(0, \cdot) \|_{H^{m_0+1}}$. Combining  (\ref{sym:etilde}),  (\ref{inid6}), and (\ref{epsilon}), we conclude
 \[ \| e^{-\kappa_0} \partial_\tau A_i(0, \cdot) \|_{H^{m_0}} \le C\epsilon.\]
 Hence it remains to estimate $\| e^{-\kappa_0}  A_i(0, \cdot) \|_{H^{m_0+1}} $, and for this  it is enough to estimate $\| \tilde{A}_i(x) \|_{H^{m_0+1}} $. From Lemma \ref{lemmaAtilde}, we conclude that $\| \tilde{A}_i(x) \|_{H^{m_0+1}}$ can be made arbitrarily small presuming  that $\| \tilde{B}_{ij}(x) \|_{H^{m_0}}$ is small enough. Hence $\tilde{E}^{1/2}_{\mathcal{SP}, m_0}(0)$ and $\tilde{E}^{1/2}_{\mathcal{VP}, m_0}(0)$ can be arbitrarily small given that we choose a small enough $\epsilon$. In conclusion,  we have a small enough $ \tilde{E}_{m_0}^{1/2}(0)$.

The remaining condition from the hypothesis of Theorem \ref{global} that we need to verify is the estimate  (\ref{eq2}). Relying on arguments similar to those of
Section 11 of \cite{Ringstrom09}, we determine  from (\ref{epsilon}) that (\ref{eq2}) holds for some $c_1 > 2$. 
Thus, noting that $\kappa_1$ only depends on $p$ and that  $m_0$ only depends on $n$, we see that Theorem \ref{global} applies here if we assume $\epsilon$ to be small enough (depending on $n$, $p$). As a consequence, we obtain a solution $(g', \psi', A')$ to (\ref{refm1}) - (\ref{refm6}) on $(t_-, \infty) \times \mathbf{T}^n$ for some $0< t_- < t_0$. It follows from  Proposition \ref{asym} that we also have the asymptotic estimates (\ref{asym1}) - (\ref{asym8}).

We verify now that the development of data on certain subsets of $B_1(0)$ solve the Einstein-Maxwell-scalar field equations. Transforming the field variables to $(g', \phi', A')$ (in accordance with the discussion in Section \ref{Reform}), we verify that these fields satisfy  the modified field equations (\ref{me1}), (\ref{me2}) and (\ref{me3}). Furthermore, since the constraints (\ref{cstrt1}) - (\ref{cstrt3}) are satisfied on $B_{15/16}(0)$ and since the initial data have been  constructed so that $\mathcal{D}_\mu (t_0, \cdot) = \partial_t \mathcal{D}_\mu (t_0, \cdot) =0$ and $ \mathcal{G} (t_0, \cdot) = \partial_t \mathcal{G} (t_0, \cdot) = 0$, it follows from standard hyperbolic PDE theory together with arguments from Section \ref{GaugeChoice} that on the domain of dependence of $B_{15/16}(0) \times \{t_0\}$,  (which we label $D[B_{15/16}(0) \times \{t_0\}]$) the functions  $(g', \phi', A')$,
satisfy  the original Einstein-Scalar-Maxwell system of equations (\ref{fe1}) - (\ref{fe2}).

Assuming $\epsilon$ to be small enough, we  conclude from Propositions \ref{geo1} and \ref{geo2} (following a line of argument similar to that found in Section 11 of \cite{Ringstrom09}) that causal geodesics starting in $B_{1/4}(0) \times \{ t_0 \} $ are future complete, and that 
\begin{equation}
B_{5/8}(0) \times  (t_-, \infty) \subseteq D[ B_{29/32}(0) \times\{t_0\}],
\end{equation}
where we increase $t_-$ if necessary. Hence we have a global-in-time spacetime which contains the development of data on a subset of the set $U\in \Sigma^n$ (from the hypothesis of Theorem \ref{maintheorem}). Within this development, causal geodesics are complete.

\vspace{4mm}
{\em Step 2: Construction of local-in-time patches, away from $U$.} To obtain a Cauchy development of the initial data on all of  $\Sigma^n$, we also need patches that are developments of initial data on subsets of $\Sigma^n$ that are not contained in $U$. Such initial data does not generally satisfy any smallness condition. 

To work with data of this sort, it is useful to define a reference spacetime metric on $\Sigma^n \times (t_-,\infty)$: We set
\begin{equation}
\label{refmetric}
\bar{g}: = (1 -f_c \circ x) (-dt^2 + h) + (f_c \circ x) (Id \times x)^* g',
\end{equation}
where $h$ is the Riemannian metric on $\Sigma^n$ included in the initial data, $g'$ is the solution metric from Step 1, $f_c$ is the cutoff function defined in Step 1, and $x$ is a coordinate map. As shown in Section 11 of \cite{Ringstrom09}, this metric  $\bar{g}$ is Lorentzian on $\Sigma \times (t_-, \infty)$.

Letting $p$ be any point in $\Sigma^n$, and choosing an open subset $S\in \Sigma^n$ which contains $p$ and has coordinate representation $\{y^1, \cdots, y^n\}$, we define coordinates 
 $\{y^0, y^1, \cdots, y^n\}$ on $\mathbf{R}\times S$ with $y^0 = t$, and we consider the PDE system
\begin{equation}
\nabla^\alpha \nabla_\alpha \phi  - V'(\phi) = 0,\quad
\nabla^\mu \bar{F}_{\mu\nu} = \nabla^\mu F_{\mu\nu} + \partial_\nu (\mathcal{G} - \bar{\mathcal{D}}_\gamma A^\gamma)= 0 ,
\end{equation}
\begin{equation}
\bar{R}_{\mu\nu}  = \partial_\mu \phi \partial_\nu \phi + \frac{2}{n-1} V(\phi) g_{\mu\nu} + F_{\mu\sigma} F_\nu\,^\sigma - \frac{1}{2(n-1)} g_{\mu\nu} F_{\rho\sigma} F^{\rho\sigma} ,
\end{equation}
where (as in Subsection \ref{RefI})
\begin{equation}
\bar{R}_{\mu\nu} := R_{\mu\nu} + \nabla_{(\mu} \bar{\mathcal{D}}_{\nu)},\quad \bar{\mathcal{D}}_\mu := \bar{\Gamma}_\mu - \Gamma_\mu,\quad \bar{\Gamma}_\mu := g_{\mu\nu} g^{\alpha\beta} \bar{\Gamma}_{\alpha\beta}^\nu,
\end{equation}
\begin{equation}
\mathcal{G} := \nabla^\mu A_\mu,\quad \bar{F}_{\mu\nu} := F_{\mu\nu} + g_{\mu\nu} (\mathcal{G} - \bar{\mathcal{D}}_\gamma A^\gamma),
\end{equation}
and $\bar{\Gamma}_{\alpha\beta}^\nu$ is the Christoffel symbol corresponding to  the reference metric $\bar{g}$.

We may choose gauges so that $\bar{\mathcal{D}}_\mu = \mathcal{G} = 0$ on $S$. Since the constraints (\ref{cstrt1}) - (\ref{cstrt3}) are satisfied on $S$, we necessarily have $\partial_t \bar{\mathcal{D}}_\mu = \partial_t\mathcal{G} = 0$ on $S$. It then follows from local existence and uniqueness results and arguments similar to those of \cite{Ringstrom08, Ringstrom09} that we obtain a patch of spacetime development $(W_p, g_p, \phi_p, A_p)$ which satisfies  (\ref{fe1}) - (\ref{fe2}); here $W_p$ is a spacetime neighborhood containing the point $p$.

\vspace{4mm}
{\em Step 3: Patching together the spacetime development patches.}

With a bit of care, as described in Section 16 of \cite{Ringstrom08}, it is straightforward to combine the development of the data related to $U$ (from Step 1), with the developments of data on sets $W_p$, for a sufficient collection of points $p$ (as per Step 2), and thereby construct 
a Cauchy development $(M, g, \phi, A)$ of the initial data on $\Sigma^n$.  This completes the proofs of Theorem \ref{maintheorem} and Theorem \ref{2ndtheorem}.

%%%%%%%%%%%%%%%%
\section{Conclusion}\label{S:conclusion}

In this paper we have proven the asymptotic GFC-stability of the inflationary cosmological model
\begin{align}
\hat{g} &= - dt^2 + (t/t_0)^{2p} \sum_{i,j=1}^{n}\delta_{ij} dx^i dx^j,\nonumber\\
\hat{\phi} &= \frac{2}{\lambda} \ln t - \frac{c_0}{\lambda} ,\nonumber\\
\hat{A}_{\mu} &= 0,\nonumber
\end{align}
as a solution to the Einstein-Maxwell-Scalar field equations
\begin{align}
R_{\mu\nu} - \frac{1}{2} R g_{\mu\nu} &= T_{\mu\nu},\nonumber\\
\nabla^\mu \nabla_\mu \phi - V'(\phi) &= 0,\nonumber\\
\nabla^\mu F_{\mu\nu} &= 0,\nonumber
\end{align}
with
\begin{align}
\nonumber p = \frac{4}{(n-1)\lambda^2} ,\quad c_0 = \ln \left[ \frac{2(np - 1) }{\lambda^2 V_0} \right],\quad V(\phi) = V_0 e^{-\lambda \phi},\\
\nonumber T_{\mu\nu} =  \partial_{\mu} \phi  \partial_{\nu} \phi - g_{\mu\nu} ( \frac{1}{2} g^{\rho\sigma} \partial_{\rho} \phi \partial_{\sigma} \phi + V(\phi) ) + ( F_{\mu\sigma} F_\nu\,^\sigma - \frac{1}{4} g_{\mu\nu} F_{\rho \sigma} F^{\rho \sigma}).
\end{align}

There are other model solutions of similar field equation systems  for which one might be able to obtain similar conclusions. These include spacetime solutions of field equations including charged scalar fields or multiple scalar fields. Along with considering various fields coupled to the Einstein equations, one might  also consider the spacetimes  with non-flat spatial slices.

\section*{Acknowledgements}
This work, which was carried out during the course of XL's PhD studies, was partially supported by NSF grant  PHY-0968612 at the University of Oregon.

\end{document}